\documentclass[12pt,english]{article}
\usepackage{mathptmx}
\usepackage[T1]{fontenc}
\usepackage{textcomp}
\usepackage[latin9]{inputenc}
\usepackage{array}
\usepackage{float}
\usepackage{url}
\usepackage{multirow}
\usepackage{amsmath}
\usepackage{amsthm}
\usepackage{amssymb}
\usepackage{graphicx}
\usepackage{geometry}
\usepackage{setspace}
\usepackage[authoryear]{natbib}
\PassOptionsToPackage{normalem}{ulem}
\usepackage{ulem}
\makeatletter

\newcommand*\LyXbar{\rule[0.585ex]{1.2em}{0.25pt}}
\providecommand{\tabularnewline}{\\}

\theoremstyle{definition}
\newtheorem{defn}{\protect\definitionname}
\theoremstyle{plain}
\newtheorem{thm}{\protect\theoremname}
\theoremstyle{plain}
\newtheorem{assumption}{\protect\assumptionname}

\usepackage{threeparttable}
\usepackage{rotating}
\newcommand\blfootnote[1]{%
\begingroup
\renewcommand\thefootnote{}\footnote{#1}%
\addtocounter{footnote}{-1}%
\endgroup
}

\newcommand{\white}{white}

\makeatother

\usepackage{babel}
\providecommand{\assumptionname}{Assumption}
\providecommand{\definitionname}{Definition}
\providecommand{\theoremname}{Theorem}

\begin{document}
\title{\vspace{-1cm}
{\large DOWN TO THE LAST STRIKE: }\\
{\large THE EFFECT OF THE JURY LOTTERY ON CRIMINAL CONVICTIONS}}
\author{Scott Kostyshak{*} and Neel U. Sukhatme{*}{*}}
\maketitle
\begin{center}
\par\end{center}
\begin{abstract}
How much does luck matter to a criminal defendant in a jury trial?
We use rich data on jury selection to causally estimate how parties
who are randomly assigned a less favorable jury (as proxied by whether
their attorneys exhaust their peremptory strikes) fare at trial. Our
novel identification strategy is unique in that it captures variation
in juror predisposition coming from variables unobserved by the econometrician
but observed by attorneys. We find that criminal defendants who lose
the ``jury lottery'' are more likely to be convicted than their
similarly-situated counterparts, with a significant increase (\textasciitilde 19
percentage points) for Black defendants. Our results suggest that
a considerable number of cases would result in different verdicts
if retried with new (counterfactual) random draws of the jury pool,
raising concerns about the variance of justice in the criminal legal
system. JEL Codes:\ H10, K41, K14.

\noindent\blfootnote{{*}Assistant Professor, Dept.\ of Economics,
University of Florida; Visiting Professor, Dept.\ of Economics and
Business, Universitat Pompeu Fabra\textemdash scott.kostyshak@upf.edu,
+34 93 542 20 00}

\noindent\blfootnote{{*}{*}Professor, Georgetown University Law
Center; Associate Dean for Research and Academic Programs; Affiliated
Faculty, Georgetown McCourt School of Public Policy\textemdash neel.sukhatme@law.georgetown.edu,
+1 202-662-4035}

\noindent\blfootnote{ (first version: November, 2018). The authors
thank participants at the 2019 American Law and Economics Association
Annual Meeting (NYU Law School), the 2019 Conference on Empirical
Legal Studies (Claremont-McKenna College), the 2020 Armenian Economic
Association Annual Meetings, Boston University Law and Economics Workshop,
the Chicago-Kent Law Faculty Colloquium, the University of Chicago
Law and Economics Workshop, the University of Florida Department of
Economics Workshop, the Georgetown Law Summer and Spring Faculty Workshops,
the Georgetown McCourt School of Public Policy Faculty Workshop, the
University of Missouri Department of Economics Workshop, the Wake
Forest University Department of Economics Workshop, the 2024 Simposio
de la Asociación Española de Economía, and the 2025 Catalan Economic
Society Congress. The authors also thank David Abrams, George Akerlof,
Michael Bailey, Rafael Calel, Adam Chilton, Tom DeLeire, Ada Ferrer-i-Carbonell,
Fran Flanagan, Michael Frakes, Brian Galle, Jonah Gelbach, Jacob Goldin,
Anna Harvey, Todd Henderson, Randi Hjalmarsson, William Hubbard, Keith
Hylton, David Hyman, Vardges Levonyan, Mike Meurer, Julie O'Sullivan,
Anna Offit, Xavier Ramos, Kyle Rozema, Steve Salop, David Sappington,
David Schwartz, Mike Seidman, Abbe Smith, Josh Teitelbaum, and Kathy
Zeiler for helpful comments. Finally, special thanks to Alex Billy,
Sara Burriesci, Flora Feng, Jeremy McCabe, Andrea Muto, Ariel Polani,
Arturo Romero Yáñez, Takayuki Sasaki, Douglas Turner, and Jonathan
Zimmer for excellent research assistance. Any errors are the authors'
own.}
\end{abstract}

\section{Introduction}

\begin{quote}
If it is a criminal case, or even a civil one, it is not the law alone
or the facts that determine the results. Always the element of luck
and chance looms large. A jury of twelve men is watching not only
the evidence but the attitude of each lawyer, and the parties involved,
in all their moves. Every step is fraught with doubt, if not mystery.

\LyXbar Clarence Darrow, ``How to Pick A Jury,'' \emph{Esquire}
\textendash{} May 1936
\end{quote}
For many people, an ideal of the criminal justice system is that like
cases should be judged equally. That a person's guilt or innocence
should not depend on certain personal characteristics, particularly
race or gender, is central to the notion of equal protection, enshrined
in the Fourteenth Amendment to the United States Constitution. More
recently, the creation of the U.S.\ Sentencing Guidelines was spurred
in large part by perceived disparities across judges in sentencing
terms awarded to similarly-situated defendants.

Juries also play a central role in determining whether criminal defendants
are treated equally, as their role as decisionmakers in criminal trials
is guaranteed by the Sixth Amendment to the U.S.\ Constitution. A
primary mechanism for creating a fair and impartial jury is randomization\textemdash from
a set of eligible jurors, a random subset is chosen to serve as the
``jury pool'' for a particular case. And from this pool, attorneys
for both sides and the presiding judge help select who will be the
final jurors for the case (i.e., the ``jury box''). Attorneys accomplish
this in part by exercising peremptory strikes, which allow a party
to exclude from the final jury a limited number of potential jurors
in the pool.

These technical challenges have limited researchers when exploring
the importance of the jury in shaping case outcomes. As some countries
are moving away from trial by jury while others are incorporating
more of it in their judicial systems,\footnote{\url{https://www.economist.com/international/2009/02/12/the-jury-is-out}}
policy makers would benefit from additional research highlighting
the variation in outcomes due to juries.

A number of recent empirical papers have used this randomization to
measure how natural variation in the jury pool affects case outcomes.
A common approach has been to use ordinary least squares (OLS) to
regress case outcomes on statistics that capture various pool demographics.
Most notably, \citet{anwar2012impact} find that juries created from
all-\white{} jury pools are significantly more likely to convict
Black defendants than juries formed from pools that contain even one
Black person, whether or not that person is not seated on the final
jury. \citet{anwar2014role} conduct a similar analysis for age, arguing
that increasing the number of older jury pool members raises conviction
rates. More recent papers by \citet{hs2021} and \citet{doi:10.1086/698193}
similarly rely on pool variation to find, respectively, that broader
gender and racial variation in the jury pool can affect conviction
rates.

While undoubtedly pathbreaking, these prior OLS-based approaches are
prone to measurement error because they rely on randomization in the
jury pool, but only a subset of pool members are ever seated on the
final jury. Some papers, such as \citet{anwar2014role}, attempt to
correct for this by using instrumental variables (IV). Specifically,
they estimate the proportion of final seated jurors that share a certain
demographic characteristic based on the proportion of initial pool
members with that characteristic.\footnote{More recently, \citet{hs2021} exploit the randomness of the ordering
of jurors in the pool rather than just using the pool composition
as an instrument. Exploiting the ordering allows for the authors to
\textquotedblleft {[}isolate{]} the as-good-as-random variation in
the gender composition of seated juries.\textquotedblright{}} For this approach to yield a consistent estimator, however, one must
also satisfy a strong assumption: the proportion of pool members who
share a demographic characteristic must affect case outcomes \emph{only}
by influencing the proportion of final seated jurors with that characteristic.

The model of \citet{anwar2012impact} walks through a specific case
where the assumption required for IV might be violated. Suppose a
prosecutor believes Black jurors favor Black defendants more than
\white{} jurors do. The prosecutor might then use some of his limited
peremptory strikes to remove Black jurors, thereby preventing him
from using those strikes on other non-Black jurors who he perceives
favor the defendant. As such, even if the final jury box does not
include a Black juror, it will be more favorable for a Black defendant
than if the initial pool had no Black members. In this scenario, the
proportion of Black jurors in the initial jury pool correlates with
the predisposition of \emph{non-Black} seated jurors, which in turn
affects the case outcome.\footnote{\citet{anwar2012impact} themselves recognize this indirect channel
would violate the exclusion restriction of IV in this scenario, noting
that IV requires a ``strong assumption that the only channel through
which the presence of Blacks in the jury pool affects trial outcomes
is by increasing the likelihood of having Blacks on the seated jury.
If, on the other hand, any of the indirect channels are important,
the IV estimates do not have a clear interpretation . . . .''}

More fundamentally, all of the identification strategies in the prior
literature are limited to:~(a) capturing differences in explicitly
specified variables; (b) that are observable to the econometrician;
(c) when the initial jury pool is created. Primarily, this includes
demographic information, such as the racial, gender, and age composition
of the jury pool. These prior approaches ignore information revealed
to the litigants but not directly measured in the data. And they are
unable to exploit relevant information revealed after the jury pool
is created. Most notably, this includes information gleaned by attorneys
during voir dire,\footnote{Information revealed during voir dire might include, for example,
a prospective juror's occupation, whether he is married, has children,
wears glasses, dresses in a suit, acts annoyed, has a strong build,
or avoids eye contact during questioning.} the process by which attorneys question and actively interact with
potential jurors in order to identify precisely what is of interest
here\textemdash juror predisposition.

Our paper surmounts these prior limitations by introducing a new identification
strategy that uses information revealed from attorneys' use of peremptory
strikes to capture how random variation in juror predispositions affects
case outcomes. Our key insight is that peremptory strikes contain
information on how favorably a litigant views the randomly assigned
jury pool. Specifically, we use strikes as proxies for hidden information
on juror predispositions that is observable to attorneys during jury
selection but is otherwise unobservable to the econometrician. As
compared to prior identification strategies, our approach captures
how case outcomes depend on random variation in jury pools that is
otherwise unobservable to the econometrician, but is observable to
attorneys during jury selection. In addition, our identification strategy
does not rely on the exclusion restriction required by IV.

Of course, a litigant's use of peremptory strikes not only reveals
his view of the initial jury pool but shapes the final jury as well.
To account for this, we focus primarily on cases at or just below
the peremptory strike limit\textemdash the maximum number of strikes
that state law permits a party to exercise in a case. By focusing
on the strike limit, we can distinguish those cases in which a litigant
might have wanted to strike more jurors but could not do so (because
she ran out of her $n$ strikes) versus those cases in which the litigant
chose to stop just one strike short of the limit (i.e., she used only
$n-1$ strikes). The identifying assumption is that these two sets
of cases are on average the same except for underlying differences
in jury composition.

We apply our identification strategy to data from a large, racially
diverse county in Florida, for all non-capital felony and misdemeanor
jury trials between January 2015 and September 2017. Our results reveal
that random variation in jury pool composition significantly impacts
outcomes for defendants, particularly Black defendants. We find that
defendants who are assigned a jury pool for which they exhaust their
peremptory strikes fare significantly worse in jury trials than similarly-situated
defendants who use one less strike than the limit. This result is
driven by Black defendants\textemdash for them, strike exhaustion
raises the chances of conviction by approximately 19 percentage points.
As such, Black defendants in our sample are subject to more variation
in the ``jury lottery'' in terms of the type of jury they might
draw. We also provide reasons why our estimates are, if anything,
likely to be a lower bound on the effect of jury composition on conviction
rates.

Because our identification rests critically on the assumption that
strike exhaustion is uncorrelated with non-juror determinants, we
test this assumption using a rich set of observables. These include,
for example, adding controls for other strikes used by parties (e.g.,
for cause strikes, peremptory strikes by the opposing party, and a
dummy for whether the opposing party exhausted its strikes); total
counts charged; defendant demographics; prior imprisonment history;
attorney experience and education; type of offense charged; type and
strength of prosecutorial evidence in a case; fixed effects for year
of jury selection and presiding judge; and observable attributes of
the jury pool, including juror race, gender, age, estimated income,
and political affiliation. The coefficient of interest remains remarkably
stable and significant across most specifications.

We also perform a number of robustness checks to indirectly test the
assumption regarding unobservables. For example, we test among different
subsets of our population that are especially likely to be similar,
cross-check with different measures of defendant guilt, compare different
measures of strike exhaustion, and conduct placebo tests using cases
in which parties do not reach the strike limit. Our results remain
robust across these different specifications.

Our paper provides at least two new contributions to the existing
literature. First, our new identification strategy allows us to identify
how variation in juror predisposition affects case outcomes without
observing this predisposition directly. Previous approaches required
one to pre-specify what observable juror characteristics (e.g., race,
age, or gender) might affect juror predisposition; our identification
strategy captures this variation as well as differences \emph{within}
observable groups. Indeed, our approach captures all feasible information
on juror predispositions, as interpreted by attorneys and revealed
during questioning, thereby uniquely positioning us to estimate the
unconditional effect of the jury lottery on case outcome to the fullest
extent possible.\footnote{Our proxy variable does not capture juror predispositions that are
unobservable to the attorneys and are not revealed through questioning,
but this variation seems unlikely to be captured by any method.} 

Second, our paper is the first to provide causal evidence on how
peremptory strikes\textemdash and the limits placed on those strikes\textemdash affect
conviction rates.\footnote{Previous work has explored the effect of peremptory strikes on jury
composition and conviction rates from a theoretical perspective. See,
e.g., \citet{moro2024exclusion}.} This is important from a policy perspective: while all states permit
some form of peremptory strikes in criminal cases, there is significant
heterogeneity among courts in terms of the procedures they follow.
We hope our research will inform ongoing policy debates as to the
merits and demerits of these differing approaches. In particular,
our results suggest that contrary to conventional wisdom, increasing
peremptory strike limits (for both sides) could benefit defendants,
particularly Black defendants, by decreasing the variance in outcomes
for similarly-situated individuals.

The rest of the article proceeds as follows. The next section of the
paper discusses the relevant prior literature. Section~\ref{sec:Identification}
provides more details on our identification strategy. Section~\ref{sec:Background}
briefly lays out how jury selection and peremptory strikes are used
in practice and describes our data. Section~\ref{sec:Results} presents
our main results, and Section~\ref{sec:Robustness} provides numerous
robustness checks. Section~\ref{sec:Policy and Conclusion} briefly
lays out potential policy implications of our study and concludes.

\subsection{Related Literature}

An early analysis of jury selection and peremptory strikes was conducted
by \citet{zeisel1977effect}, who worked with a federal district court
in Illinois to create ``mock juries'' composed of struck and unused
jurors from jury selection conducted in 12 actual criminal trials.
The mock jurors were then presented with abridged facts for the cases
they would have adjudicated and were asked to render a verdict. Among
other things, the authors found that peremptory challenges appeared
to change the verdict for at least one case. Other studies, as summarized
in \citet{devine2001jury}, have also used mock juries to study how
jury composition affects outcomes.\footnote{\citet{baldus2001use} analyzed how peremptory strikes were used in
317 capital murder cases tried by jury in Philadelphia in the 1980s
and 1990s. \citet{diamond2009achieving} find that peremptory strike
use was related to juror race/ethnicity in 277 civil jury trials,
but that parties' use of strikes tended to cancel one another out.
A few other papers, such as \citet{flanagan2015peremptory} and \citet{ford2009modeling},
created formal models of the peremptory strike process.}

More recently, a seminal paper by \citet{anwar2012impact} uses variation
in people randomly assigned to a jury pool to measure how the racial
composition of the jury pool affects trial outcomes. Using data from
two rural Florida counties (Sarasota and Lake), they find that juries
formed from all-\white{} jury pools are 16 percentage points more
likely to convict Black defendants relative to \white{} defendants.
This effect is eliminated if the pool happens to include at least
one Black person, even if that person is not seated on the final jury.\footnote{See also \citet{alesina2014test}, who find higher reversal rates
for minority defendants who were convicted of killing \white{} victims
in Southern states, suggesting bias at the trial level. Other recent
empirical papers that explore the impact of race within the criminal
justice system include \citet{rehavi2014racial} (prosecutor charging
decisions), \citet{abrams2012judges} (likelihood of incarceration),
and \citet{arnold2018racial} (bail decisions).} 

\citet{anwar2014role} build on this work by using the same quasi-random
variation in jury pools to measure how the age of jury pool members
can affect conviction rates. The paper presents both reduced form
(direct effect of pool variation on outcomes) and IV (effect of variation
in the final seated jury on outcomes) estimates that show a strong
effect, whereby defendants randomly assigned to ``older'' pools
are much more likely to be convicted than those randomly assigned
to ``younger'' pools. They also find that prosecutors used peremptory
challenges to remove younger members in the jury pool, while defense
attorneys used such challenges to remove older members.

More recent papers apply the same identification strategy to measure
how variation in other characteristics of jury pool members affects
case outcomes. \citet{doi:10.1086/698193} applies these OLS and IV
techniques to North Carolina data, finding that the presence of more
\white{} jurors in a pool is associated with higher conviction rates,
and that prosecutors tend to strike potential Black jurors and defense
attorneys tend to strike potential \white{} jurors. \citet{hs2021},
using data from two other Florida counties (Palm Beach and Hillsborough),
exploit the random pool assignment and, in a novel contribution, the
random \emph{ordering} of jurors in the pool to construct a weighted
average of the characteristic of interest\textemdash in their case,
whether a juror has the same gender as the defendant. The weight for
each juror in the pool is the estimated unconditional probability
that a juror in that position is selected into the box. They find
that increasing the weighted average by one standard deviation (\textasciitilde 10
percentage points) reduces conviction rates on drug charges by 18
percentage points. Finally, \citet{anwar2022unequal} show that Black
defendants receive longer sentences due to over-representation of
jurors from predominantly \white{} and high-income neighborhoods.

\section{Identification Strategy\label{sec:Identification}}

In this section, we present and discuss the assumptions required for
our proxy variable (i.e., peremptory strike exhaustion) to yield a
consistent estimator  of the effect of the jury lottery on criminal
case outcomes. The first link in our identification chain recognizes
that differences in average conviction rates for defendants who use
all of their peremptory strikes ($n$ strike group) and those who
use one less strike than the limit ($n-1$ strike group) are most
closely a measure of attorney beliefs regarding the favorability of
a particular jury pool. The next link in the chain is from attorney
beliefs to reality. It seems reasonable that attorneys are best positioned
to assess the predisposition of potential jurors. Further, attorneys
have strong incentives (greater than, for example, a judge) to develop
skills to assess juror predisposition, given their role as advocates
and repeat players in courts.

Our identification strategy solves an unobserved variable problem.
Ideally, we would like to observe in some direct way whether potential
jury members are more favorable to the prosecution or the defense,
as compared to the average jury member. No such data contain this
information, and it is difficult to even imagine an ethical and feasible
method to capture such information directly. By measuring juror predispositions
through the lens of a litigant, our proxy variable picks up as much
variation as feasibly possible given the data available. 

Our basic identification strategy is driven by the subsample of cases
in the $n$ and $n-1$ categories. Without valid extrapolation, these
results might be interesting only to the extent this category represents
a considerable subpopulation. In fact, 61.6\% of our cases involved
parties who used either $n$ or $n-1$ strikes (defendants used $n$
or $n-1$ strikes in 54.7\% of cases and prosecutors used $n$ or
$n-1$ strikes in 38.4\% of cases).\footnote{The prevalence of the $n$ and $n-1$ groups is not unique to our
data. For example, tabulating North Carolina felony jury data from
2010\textendash 2012 as reported in \citet{doi:10.1086/698193}, we
can see prosecutors and defendants used either $n$ or $n-1$ strikes
in 17.4\% and 41.0\% of cases, respectively.}

The main assumption needed for identification is that there is no
difference on average between $n$ strike trials and $n-1$ strike
trials that affects case outcomes except for differences in jury composition.
In practice we just need this assumption to hold conditional on non-jury-composition
observables, since we test the sensitivity of our results in various
conditional specifications.   For example, one way in which the
assumption might be violated would be if attorney quality is correlated
with the number of strikes used; then the variation of trial outcomes
between $n$ and $n-1$ cases would pick up the effect of attorney
quality. To address this possibility, we include various controls
for attorney quality, such as practice experience, ranking of law
school attended, use of strikes in other cases, and whether the defense
attorney is a public defender.

Based on the framework discussed above, we can make an intuitive prediction:
we expect a party's win rate to be higher for cases in which $n-1$
strikes are used, as compared to cases in which $n$ strikes are used.
The magnitude of this predicted effect, and whether it impacts prosecutors
and defendants differently, likely depends on the underlying distribution
of jurors. By exploring this prediction and estimating the difference
between the two strike groups, we can estimate how random variation
in jury composition affects case outcomes.\footnote{We formalize this intuition in a model in the Appendix.}

\section{Background on Jury Selection and Data\label{sec:Background}}

\subsection{Background on jury selection}

Applying our identification strategy to the data requires a deeper
understanding of how jury selection works in Florida, which conducts
the process in a manner typical of many states. The state maintains
a list of potential jurors who may be summoned, based on individuals
who have received a drivers' license or state identification card.
Eligible jurors must be U.S. citizens, at least 18 years of age, who
are not convicted of a felony and are residents of the county in which
they are to be summoned.\footnote{Some jurors may request to be excused from participation for legitimate
reasons, such as if they have already been summoned and reported within
the last year, they are above 70 years of age, they have a medical
condition that makes them unable to serve, they are an expectant mother,
or they are not employed full-time and are a parent of a child less
than 6 years of age. \emph{See} Fl. St. § 40.013. A large number of
jurors also fail to show up for their summons and do not provide an
excuse. These jurors might be summoned again; some jurisdictions in
the United States even pursue criminal charges against such jurors
if they repeatedly do not show up when summoned, though such charges
are rare.}

If a jury trial is anticipated on a particular date, a jury administrator
will summon a certain number of potential jurors from an eligible
list of jurors for that week. The administrator then picks a random
subset of the jurors who appear for jury duty to create pools for
particular cases. These groups of jurors, ordered by a randomly assigned
number, are then sent to a specific courtroom, where jury selection
can begin.\footnote{Following \citet{anwar2012impact}, we test whether jury pools are
actually randomly assigned in our jurisdiction by regressing certain
observable pool characteristics\textemdash the proportion of female
and Black jurors, average juror age, median juror income (calculated
based on the median income in the zip code in which a juror lives),
proportion of registered Democrats, and proportion of registered Republicans\textemdash on
various defendant, attorney, and case characteristics. Our results
are shown in Appendix Table \ref{fig:Randomization check}. Out of
the 96 pair-wise comparisons in the table, there are nine that are
significant at the 5\% significance level or less. The magnitudes
of all coefficients are very small. While a joint F-test suggests
as a whole this specification might have one or more coefficients
that differ from 0, all of the estimated coefficients are very small.
Moreover, including these summary jury pool measures in our regression
specifications does not affect our results. On the whole, these results
support an inference that the jury pools were indeed randomly constructed.}

The judge then begins voir dire\textemdash the process of asking potential
jurors questions to whittle down the jury pool into the final jury.
The judge first explains some basic aspects of the case to jurors
and then asks each juror a standard series of questions that pertain
to their ability to be fair and impartial during the trial. The prosecuting
and defense attorneys can then ask follow-up questions for particular
jurors. Both sets of attorneys are present for the entire questioning.

Typically the pool then leaves the room while the attorneys propose
who should be struck from the pool for cause. These are individuals
who, based on their answers during jury selection, an attorney argues
will be unable to serve fairly and impartially at trial. The judge
may agree to strike these individuals and may excuse others at her
own discretion as well, if she finds they are unable to serve impartially
or believes jury service would impose a considerable hardship on them.
In the sample, 30.8\%, 2.7\%, and 2.3\% of all potential jurors in
jury pools were struck for cause by the judge, prosecution, or defense,
respectively. 

Next, jury selection proceeds to peremptory strikes. Unlike for cause
strikes, an attorney need not provide any reason why a juror should
be struck when exercising his peremptory strikes. The only requirement
is that parties cannot strike based on the juror's race, gender, or
national origin.\footnote{See \emph{Batson v.\ Kentucky}, 476 U.S.\ 79 (1986) (race); see
also \emph{J.E.B.\ v.\ Alabama ex rel.\ T.B.}, 511 U.S.\ 127 (1994)
(gender). If a party suspects the opposing attorney is using his peremptory
strikes in an impermissibly discriminatory manner, she may challenge
a strike as it is made. The attorney who made the peremptory strike
must then proffer a race- and gender-neutral reason for it. The party
that issued the \emph{Batson} challenge can then counter the striking
party's explanation and explain why it is pretextual. The judge immediately
decides whether to allow or disallow the peremptory strike; if she
does the latter, then the attorney does not lose that peremptory strike
and can use it to strike another juror.}

If a juror is struck either for cause or via a peremptory challenge,
the next juror in line will take his place. Both the prosecution and
the defense know who this person is before they exercise their strike,
and have the same type of information on the potential replacement
as they do on the person they might strike, since all potential jurors
are questioned in front of both parties before peremptory strikes
begin.

Both sides are provided the same number of peremptory strikes, with
the strike limit set by the type of case. In Florida, parties receive
3 peremptory strikes for misdemeanors (maximum punishment of up to
one year imprisonment), 6 strikes for most felonies (maximum punishment
of at least one year but less than life imprisonment), and 10 strikes
for cases involving charged felonies in which a defendant might receive
life imprisonment. In our data, 9.6\% and 11.8\% of all potential
jurors were struck by peremptory challenges from the prosecution or
the defense, respectively.

During the peremptory strike process, the parties start with the first
remaining jury candidate and proceed one by one through the pool.
First the prosecution decides whether to strike a candidate, then
the defense. Once both parties have exhausted their strikes or affirmatively
declined to use all of their strikes, then the final jury is set and
comprises the first 6 people in the jury pool who have not yet been
struck.\footnote{Most states use 12-person juries for most crimes. Florida differs
in that it uses 6-person juries for all crimes except capital offenses,
for which it uses 12-person juries. We exclude all capital cases in
our analysis.} Throughout this process, the jurors are typically not present and
are not informed why any particular juror was struck or by whom she
was struck.

Table \ref{tab:Voir-Dire-Strike-Sheet} illustrates how voir dire
might play out in a hypothetical non-life-eligible felony case. In
this example, jurors 1, 2, 7, 8, 13 and 20 were struck by the judge
for cause, jurors 11 and 24 were struck by the defense for cause,
and jurors 15 and 27 were struck by the prosecution for cause. The
prosecution and defense then exercised their peremptory strikes, with
the defense striking jurors 4 and 5, the prosecution striking juror
9, and so on. In the end, the defendant exhausted all 6 of his peremptory
strikes, whereas the prosecutor only used 4 of her strikes. The final
6-person jury comprises jurors 3, 6, 10, 12, 21 and 25. In our data,
22.9\% of jury pool members were neither struck nor used, and 20.0\%
ended up on the final jury.

\subsection{Summary statistics for strike groups}

Our data come from a large, racially diverse county in Florida, and
they comprise all non-capital felony and misdemeanor cases in which
jury selection was conducted from January 2015 through September 2017.
The data include detailed information on the trial participants, including
the name, race, gender, age, and address of the defendant, and the
names of the presiding judge and attorneys in the case. Using data
from the Florida Department of Corrections, we can also control for
a defendant's criminal record by measuring the number of times he
has previously been imprisoned in Florida state prison, and for how
long. Using Florida state bar records, we supplement these data with
information on the attorneys' practice experience and educational
background. We also have data on which charges were brought against
each defendant under which statutory provision. Our data further include
information on the potential jurors in each case, including their
demographic information, their juror number, and whether they were
struck for cause (by the judge, prosecution, or defense), struck under
a peremptory challenge (by the prosecution or the defense), not used
in voir dire, or seated on the final jury.

We have a total of 567 cases involving a single defendant\footnote{We exclude cases in which multiple defendants are tried jointly, since
peremptory strike limits are different in those cases. We also exclude
cases in which voir dire was conducted multiple times.} for which a jury was selected and neither party exceeded the peremptory
strike limit.\footnote{We exclude 37 cases in which one or both parties exceeded the peremptory
strike limit. As discussed below, these are likely the rare instances
in which a judge raised the limit.} Out of these cases, the parties tried the case to a jury verdict
on 511 occasions (90.1\% of the time). To the extent prosecutors and
defendants can identify whether a particular jury is good or bad for
them, we should expect jury selection to affect both cases that are
tried to a verdict and cases that settle after voir dire is conducted.
As such, in our primary specifications, we look at all case outcomes
(not just cases that proceed to jury verdicts). Nonetheless, and as
discussed below, we also test numerous specifications in which we
limit our sample to cases in which a jury actually rendered a verdict,
and we find our results remain largely similar.

We have a total of 213 misdemeanors, 273 non-life-eligible felonies,
and 81 life-eligible felonies. They span the full range of criminal
offenses commonly prosecuted in state court, including driving under
the influence, battery, burglary, robbery, theft, drug offenses, sex
offenses, and murder. The defendants are 55.38\% Black, 31.39\% \white{},
12.52\% Hispanic, and 0.35\% each Asian and Indian.

Our identification strategy is similar to regression discontinuity
in that we are using strike limits as a way to distinguish jury pools
with exogenous variation in how pro-prosecution or pro-defense they
are. Our treatment groups in most specifications are parties who exhaust
their peremptory strikes ($n$ strikes); our control groups are parties
who use one less strike than the limit ($n-1$ strikes), though we
test other control groups (e.g., parties who used $<n$ strikes),
as discussed below.

Figure \ref{fig:Peremptory-Strikes:-Defendant} shows the distribution
of cases by number of peremptory strikes used by defendants across
different offense categories and for all cases in our dataset. For
both non-life-eligible felonies and misdemeanors, the two strike categories
with the most cases are the $n$ and $n-1$ categories, with fewer
cases in categories with fewer strikes. For life-eligible felonies,
the distribution is shifted to the left and is centered on 6 strikes
(\emph{$n-4$}). Across all offense categories, the defendant exhausted
all of his strikes in 31.1\% of the cases, and used one less strike
than the limit in 20.9\% of the cases. These are the two most prevalent
outcomes.

One can see the number of peremptory strikes used drops sharply above
the limit for misdemeanors and non-life-eligible felonies, providing
suggestive evidence that this limit is a binding constraint on defendants
in a large number of cases in those offense categories.\footnote{By contrast, we see no such pattern for other, non-peremptory strikes.
For example, Figures \ref{fig:For Cause-Strikes:-Judge}, \ref{fig:For Cause-Strikes:-Defense},
and \ref{fig:For Cause-Strikes:-Prosecution-1} in the Appendix show
the distribution of judge-issued, defense, and prosecution for cause
strikes, respectively, across the three broad offense categories.
The spread of such strikes is much larger than in the peremptory strike
context, with no discernible cutoff (since there is no limit to how
many of these strikes can be granted).} The few cases (5.0\%) with more strikes than the limit are the rare
situations in which the presiding judge allowed defendants to use
more than their allotted share (something we confirmed is uncommon
through conversations with practicing attorneys in the jurisdiction).

Figure \ref{fig:Peremptory-Strikes:-Prosecution} shows an analogous
distribution of cases by number of prosecution peremptory strikes.
While the overall shape of the distribution is similar, one can see
the modal number of strikes for prosecutors falls to $n-1$ for misdemeanors,
$n-3$ for non-life-eligible felonies, and $n-6$ for life-eligible
felonies. Overall, prosecutors use fewer peremptory strikes than defense
attorneys; as such, the data suggest that peremptory strike limits
are less likely to serve as binding constraints on their choice of
jurors. Across all offense categories, prosecutors exhausted all of
their strikes in 17.1\% of the cases, and used one less strike than
the limit in 21.4\% of the cases. They exceeded the strike limit in
1.8\% of cases.

We now examine whether various baseline case-specific characteristics
vary across prosecution and defense strike classes. Table \ref{tab:Summary-Statistics}
shows mean values within strike classes for variables that capture
defendant demographics and prior criminal history; attorney experience
and education; and jury pool characteristics. Table \ref{tab:Summary-Statistics-1}
does the same for other case-specific characteristics, including which
offenses are charged and the types of evidence against a defendant
in a case. Both tables include results for two-sided t-tests for differences
in means between the $n$ and $n-1$ groups, as well as data on all
cases in our sample and on cases with $n-2$ strikes (i.e., those
in which a party used two less strikes than the limit).

We can see that broadly speaking, criminal defendants, attorneys,
and case and jury pool characteristics\footnote{One difference is that prosecutors appear to have slightly more registered
Republicans in the jury pools from which they exhaust their strikes
as compared to those pools in which they use one less strike than
the limit. This should not impact cases in which defendants use $n$
versus $n-1$ strikes.} appear comparable across different strike groups. There are no statistically
significant differences in terms of defendant age, or how likely a
defendant is to be Black, Hispanic,\footnote{Because the data reported most ethnically Hispanic defendants as \white{},
we use a standard R-package, ethnicolr, that relies on Florida voter
registration and Wikipedia data to predict race and ethnicity based
on first and last name. See https://github.com/appeler/ethnicolr.
Following \citet{arnold2018racial}, we chose a cutoff of 0.7 predicted
probability of being Hispanic in determining whether to categorize
a \white{} defendant as such. Our results do not depend substantially
on which threshold we use to determine whether a defendant is Hispanic.} or female across the $n$ and $n-1$ strike groups. Similarly, there
are no statistically significant differences across the $n$ and $n-1$
strike groups in terms of how many prison stints a defendant has previously
served, though $n$ strike defendants have served on average 0.6 years
more prior years in state prison relative to the $n-1$ strike group
(significant at the 10\% level). Defendants in both strike classes
face similar numbers of counts in the current case against them.

Regarding attorney characteristics, defense attorneys who use up their
strikes appear to have attended slightly better ranked law schools\footnote{Rankings are based on the 2018 U.S. News \& World Report Rankings,
commonly used to rank law school.} than their counterparts who had one strike remaining, and they face
slightly less experienced prosecutors. To the extent law school ranking
or attorney experience correlates with attorney quality, we might
expect this to downward bias our estimate of the impact of defendant
strike exhaustion on case outcomes. Otherwise, there are no statistically
significant differences in terms of defense or prosecutor experience,
whether the defense attorney is a public defender, or in terms of
the ranking of the prosecutor's law school. At any rate, we control
for these and other observable characteristics in our regressions
below and find our results to be largely unchanged.

Turning to Table~\ref{tab:Summary-Statistics-1}, there are no statistically
significant differences between the $n$ and $n-1$ strike classes
at the 5\% level in terms of the relative prevalence of common classes
of cases.\footnote{These are the same offense classes as used in \citet{anwar2012impact}
and subsequent papers. At the 10\% significant level, it appears that
defendants who exhaust their strikes are marginally more (less) likely
to be charged with violent offenses other than murder (property offenses),
and prosecutors who exhaust their strikes are marginally more likely
to be involved in cases involving other offenses.} Figures \ref{fig:Def-Crime-Peremptory-Strikes} and \ref{fig:Pl-Crime-Peremptory-Strikes-1}
show this graphically, comparing the relative frequency of these cases
across strike classes for both the prosecution and defense. Once again,
we see no large deviations in the relative prevalence of certain types
of cases in particular strike categories. These results are largely
confirmed by two-sided t-tests for all pairwise comparisons within
offense groups and across strike classes.\footnote{A two-sided t-test for all 42 pairwise comparisons between the strike
classes yields 41 insignificant comparisons at the 0.05 level, with
only the relative frequency of drug crimes significantly decreasing
between the $n-2$ and $n-1$ prosecution strike classes.}

Because even similarly charged cases might include significant heterogeneity,
we gathered more detailed information on each case by reviewing its
probable cause affidavit. Such affidavits list the alleged facts in
each case, as collected by law enforcement when an arrest is made
and as sworn before a court. Importantly, the affidavits reveal what
evidence the prosecution has, which helps determine the strength of
the case. We coded for whether the affidavits describe any of the
following types of prosecutorial evidence: (1) photos or video that
show the defendant committing the alleged offense; (2) incriminating
items that were recovered from the defendant's possession, dwelling,
vehicle, or place of business; (3) physical (forensic) evidence that
links the defendant to the offense and requires expert analysis and
testimony (e.g., DNA, fingerprints, tire or shoe tracks, ballistic
materials, trace fibers, recovered computer files); (4) documentary
evidence (e.g., bank, business, phone, email, or medical records linking
the defendant to the offense; breathalyzer or other intoxication test
results); (5) a confession by the defendant to the police, whether
orally or in writing; and (6) any admission by the defendant against
his/her interest, which might fall short of a full confession and
is not necessarily made as part of formal statement to police.\footnote{In addition, we coded whether the underlying offense took place in
a correctional setting (i.e., prison or jail) or not. This variable
does not appear to be correlated with guilty outcomes in our sample,
and including it as a control does not substantially affect the results.
We also attempted to code for the presence and identity of eyewitnesses
for each case; however, there was significant variation across different
coders because this information was difficult to gather from the affidavits.
At any rate, this variable does not appear correlated with guilty
outcomes in our sample, and including it as a control does not substantially
affect our results.} The presence of forensic evidence, recovered items, documentation,
a confession, or an admission are strongly positively correlated with
guilty outcomes for cases in our sample; the presence of surveillance
photos/videos was also positively correlated with guilty outcomes,
but below conventional levels of statistical significance.

Table \ref{tab:Summary-Statistics-1} shows that there are no statistically
significant differences between the $n$ and $n-1$ strike classes
at the 5\% level in terms of the relative prevalence of these different
kinds of prosecutorial evidence. This suggests there is no systematic
difference between these strike classes in terms of the quality of
evidence against the defendant.\footnote{At the 10\% significance level, forensic evidence appears five percentage
points less likely to appear in the $n$ strike cases relative to
the $n-1$ strike cases. This does not appear to impact our results,
however\textemdash when we include it as a control (or alternatively
exclude cases with any forensic evidence from our sample), our results
remain substantially the same.}

\section{Results\label{sec:Results}}

\subsection{Peremptory strikes and conviction rates}

We now explore how case outcomes differ when a litigant exhausts all
of her peremptory strikes (i.e., $n$ strikes used) as compared to
when she has exactly one strike remaining (i.e., $n-1$ strikes used).\footnote{Like the previous literature, we use ordinary least squares with heteroskedastic
robust standard errors. Our results remain largely robust if we cluster
at the defendant attorney level or judge level.} Table \ref{tab:Defendant-Pl-Strike-Exhaustion} presents our main
regression results. Following \citet{doi:10.1086/698193}, our outcome
variable is the proportion of charges on which a defendant was convicted\footnote{In 26 cases in the dataset, the court ``withheld adjudication''
on at least one count, which means the defendant was found guilty
but was technically not deemed convicted by the presiding judge because
the court believed he would be unlikely to recidivate. \emph{See}
Fl. St. § 948.01(2). Since our primary focus is on the effect of the
jury rather than the judge on outcomes, we treat these cases as guilty
outcomes, though our results remain substantially the same if we instead
exclude these cases from our sample.} that were yet to be decided as of the date of jury selection (i.e.,
were not dropped or otherwise adjudicated before that date). Our preferred
specification thus captures nuances in mixed outcome cases, such as
when a jury convicts on one count but acquits on another.\footnote{Our results remain robust and substantially the same if we treat a
defendant as guilty if he is convicted on at least one charged offense.
But such specifications might introduce substantial measurement error,
since they treat any mixed outcome case as a defendant loss (and mixed
outcome cases comprise 25.4\% of the dataset). To give a concrete
example, one case in the dataset involves a defendant who was found
not guilty by a jury of sexual battery and aggravated battery with
a deadly weapon, but found guilty of assault. He received no additional
prison term for the conviction (sentenced to time served). Characterizing
this case as a total loss for the defendant is arguably inaccurate
relative to our preferred approach, which categorizes this case with
a guilty value of 0.333 to account for the fact that the defendant
was found guilty on just one of the three charges he faced at the
time of jury selection.} 

Our coefficient of interest is a dummy variable for ``Exhausts Strikes.''
In columns (1)\textendash (3), this is a dummy for whether the defendant
used up all of his peremptory strikes (and = 0 otherwise); in columns
(4)\textendash (6), it is a dummy for whether prosecutor used up all
of her peremptory strikes. Similarly the control variable $n/n-1$
Group is a dummy for whether the defendant used either $n$ or $n-1$
peremptory strikes; in columns (4)\textendash (6), it is a dummy for
whether the prosecutor used either $n$ or $n-1$ peremptory strikes.
Because the number of strikes allowed differs across misdemeanors
and the two classes of felonies, all specifications also include controls
for whether the defendant was charged with a felony and whether it
was a life-eligible offense. With these controls, Exhausts Strikes
measures how outcomes differ in cases in which all peremptory strikes
were used as compared to ones in which one less strike was used than
the strike limit.\footnote{Our coefficient of interest on Exhausts Strikes remains largely similar
if we apply the same regressions separately to misdemeanors and the
two different classes of felonies, or if we exclude life-eligible
felonies from our analysis, though some results lose significance
as our decreased sample size increases our standard errors.}

Columns (2), (3), (5) and (6) also include dummy variables for whether
the defendant was Black, Hispanic, or female, and controls for the
defendant's age and the number of years he has previously served in
Florida prison.\footnote{Replacing or supplementing $\text{age}$ with $\text{age}{}^{2}$
or $\ln(\text{age})$, or replacing years of past imprisonment with
number of past imprisonments, does not meaningfully affect our results.} Columns (3) and (6) further include controls for the following case-
and strike-specific characteristics: number of prosecution-, defendant-,
and judge-directed for cause strikes (each tabulated separately);
number of peremptory strikes used by the opposing side; a dummy for
whether the opposing side exhausted its peremptory strikes; number
of counts charged; and a fixed effect for the year in which jury selection
occurred.

The coefficient on Exhausts Strikes is positive and significant at
the 5\% level in columns (1)\textendash (3). The coefficient values
indicate that strike exhaustion across all defendants increases the
probability of conviction by about 12 percentage points. By contrast,
the coefficient on Exhausts Strikes is insignificant for prosecutors.

As noted in the prior discussion on identification, if differences
in jury composition are the only relevant differences between cases
in which defendants use $n$ rather than $n-1$ strikes (at least
after controlling for observables), then these regression results
have a causal interpretation. In particular, a defendant who receives
a ``bad'' jury pool (i.e., one in which he is forced to exhaust
his peremptory strikes) is about 12 percentage points more likely
to be convicted of a crime than a similarly-situated defendant who
does not use up his strikes.

\subsection{Effects on Black and non-Black defendants}

The above results suggest that strike exhaustion by a defendant is
a better predictor of whether the defendant is found guilty than strike
exhaustion by the prosecutor. We now examine whether this effect on
defendants is homogenous or depends on a defendant's race.

Table \ref{tab:Defendant-Bl-NBl-Exhaustion} presents our results,
using the same specifications from columns (1)\textendash (3) in Table
\ref{tab:Defendant-Pl-Strike-Exhaustion} but instead divided by race
into Black and non-Black defendants. We find the probability of conviction
increases by 18\textendash 20 percentage points for Black defendants
who exhaust their peremptory strikes ($n$ strikes used) as compared
to Black defendants who use one less strike than the limit ($n-1$
strikes used). By contrast, the coefficients on peremptory strike
exhaustion for non-Black defendants are between 0.04 and 0.07 and
are statistically insignificant.

Our results indicate that Black defendants who exhaust their peremptory
strikes are more likely to be convicted than Black defendants who
use one less strike than the limit. Moreover, strike exhaustion increases
conviction rates more for Black defendants than for non-Black defendants
in our sample. Still, we cannot rule out that strike exhaustion has
the same effect on non-Black defendants as it does for Black defendants
when we generalize to the full population. In particular, when we
regress the conviction dummy variable on whether a defendant is Black,
whether a defendant expires his peremptory strikes, and the interaction
of these two variables (along with controls for whether the underlying
crime was a felony or a life-eligible crime), we obtain a positive
but insignificant coefficient for the interaction. This null result
might be due to our relatively large standard errors, which in turn
are likely caused by the relatively small sample size of our study.

\subsection{Within- versus between-group effects of jury composition}

As noted, a key feature of our novel identification strategy is that
it allows us to capture the full unconditional effect of the jury
lottery\textemdash that is, differences across observable groups as
well as differences within those groups. Previous approaches could
only capture how pre-specified racial, gender or age differences in
the randomly-assigned jury pool affect conviction rates. We build
on this by further capturing how variation within identifiable race-gender-age
groups (e.g., 30-year-old \white{} men) might affect conviction rates.

We can compare whether within-group variation or between-group variation
in prospective jurors has a greater effect on conviction rates in
our sample. To do this, we include controls for the composition of
the randomly-selected jury pool, before the parties exercise any strikes.
We measure the proportion of Black jurors in the jury pool,\footnote{\citet{anwar2012impact}, who conduct their study in largely \white{}
counties, contrast cases in which there were any Black jury pool candidates
versus cases in which there were none. Because our jurisdiction is
significantly more diverse, there are very few cases in which there
were no Black candidates in the jury pool (only 18 out of 567 total
cases). As such, we cannot test their specifications here.} the proportion of women in the jury pool, the average age of members
of the jury pool, the log median income of jury pool members (with
each juror's income estimated by the median income for the zip code
in which she resides, using U.S. Census data and 2017 inflation-adjusted
dollars)\footnote{Our results are substantially the same if we use median income instead
of its log. We present log median income as a control in the regressions
here rather than levels of median income because we expect the potential
income effect to be concave.}, and the proportion of registered Democrats and registered Republicans
in the jury pool prior to strikes.\footnote{We obtained the political affiliation of potential jurors by matching
jury pool data with Florida voter registration data using juror last
name, zip code, and date of birth. We matched 70.89\% of all jury
pool members (17,077 out of 24,088) in our sample in this manner.}

Table \ref{tab:Defendant-All-Bl-Jury-Composition} shows the same
specifications as in columns (1)\textendash (3) of Tables \ref{tab:Defendant-Pl-Strike-Exhaustion}
and \ref{tab:Defendant-Bl-NBl-Exhaustion}, except with the race,
gender, age, income, and political affiliation jury pool controls.
Most of these controls do not have a statistically significant impact
on case outcomes, with the exception of the proportion of registered
Republicans in the jury pool and the effect of log median juror income,
both of which are marginally statistically significant for Black defendants.
Using standardized regression coefficients (not shown), we can see
that increasing log median juror income (proportion of registered
Republicans) in a jury pool by one standard deviation increases conviction
rates by about 9.5 to 10.2 (10.7 to 11.6) percentage points for Black
defendants.

Importantly, the coefficient on defendant peremptory strike exhaustion
remains positive, statistically significant, and substantially the
same magnitude as in prior specifications. The enduring importance
of defendant strike exhaustion implies that within-group variation
among jurors plays a more substantial role in case outcomes than previously
recognized.

\subsection{Bounds discussion}

The large differences we find in conviction rates between the $n$
and $n-1$ groups are likely just a lower bound on how jury composition
affects case outcomes. This is true for at least two reasons. First,
our results would measure the maximum possible difference between
``unlucky'' and ``lucky'' defendants only if defendants who used
all of their $n$ strikes were all assigned the worst possible jury
pools, and defendants who used $n-1$ strikes were all assigned the
best possible jury pools. One or both of these events seems unlikely;
as such, the difference between the most unlucky and most lucky defendants
is likely greater than what we measure here.

Second, it is unlikely that every defendant within the $n-$strike
group ran out of peremptory strikes. Some of these defendants likely
wanted to strike $n$, and exactly $n$, potential jurors and as such,
the strike limit was not a binding constraint for them. Since we cannot
distinguish those $n-$strike defendants from other $n-$strike defendants
for whom the strike limit was binding (i.e., those who wanted to strike
$n+1$, $n+2$, or more potential jurors), our estimate of the impact
of jury composition is likely downward biased, with an actual impact
even greater than what we measure here.

\section{Robustness\label{sec:Robustness}}

The identifying assumption about differences between cases involving
$n$ and $n-1$ strikes forms the basis of our causal claim that variation
in jury composition heavily affects whether a defendant\textemdash particularly
a Black defendant\textemdash is convicted. In this section, we present
other specifications to test various potential threats to our identification.
In particular, we test whether variation in attorney quality, crimes
charged, strength of prosecutorial evidence, judges, or other case-specific
features might drive our results. We also test different outcome measures
and measures of strike exhaustion, and we conduct placebo tests.

\subsection{Attorney quality\label{subsec:first-story}}

Arguably the biggest concern might be that differences in attorneys
are driving the results\textemdash namely, bad defense attorneys might
be both more likely to use more peremptory strikes and more likely
to attain worse outcomes for their clients at trial. We can try to
control for any such differences through observable characteristics
of the attorneys. In particular, we include controls for attorney
experience (measured in years from their date of bar passage to the
date of jury selection), ranking of law school attended, and whether
or not the defense attorney is a public defender.\footnote{We code law schools that are unranked or whose rank is not published
with a rank of 175 (the midpoint ranking if all unranked law schools
were ranked), though our results are not sensitive to this choice.
Also, measuring law school quality instead based on the natural log
of law school ranking does not materially affect our results.}

Columns (1) and (5) of Table \ref{tab:Defendant-Strike-Atty} present
the results of including these covariates for all defendants and Black
defendants, respectively, with the coefficient of interest once again
whether the defendant exhausted his strikes in the category of cases
for which either $n$ or $n-1$ strikes were used. As one can see,
the regression results remain largely unchanged (or are perhaps even
slightly stronger) as compared to the coefficients in the baseline
specifications presented earlier.

Still, it is possible that years of practice experience and quality
of law school attended do not fully capture differences across attorneys,
and that some unobservable attorney characteristics are driving both
differences in the number of strikes used and the attorney's ability
to help her client. As a further test, we include counts for the number
of cases in our sample in which each attorney used $n$ strikes and
the number of cases in our sample in which they used $n-1$ strikes,
to capture their propensity to exhaust or nearly exhaust their strikes.
As a second test, we limit the sample to public defenders, who arguably
share similar unobservable qualities with one another. As a third
test, we limit the sample to public defenders and include the controls
for number of cases in which attorneys used $n$ and $n-1$ strikes.\footnote{As yet another test (not reported here), we limit our sample to cases
involving attorneys who used both $n$ strikes and $n-1$ strikes
at least once in our sample. Such attorneys have already shown they
are willing to exhaust or fall just short of the strike limit; as
such, one might expect them to be similar to one another on unobservable
dimensions that might bear on their decision to use strikes. Limiting
the sample in this manner reduces the coefficient's magnitude and
significance across all defendants, though it remains positive. However,
it does not appreciably change the magnitude or statistical significance
of the coefficient for Black defendants. In short, Black defendants
who are represented by the subset of attorneys who have used both
$n$ and $n-1$ strikes in the sample, are public defenders, and have
comparable practice experience and legal pedigree, are more likely
to be convicted of a crime as compared to their Black defendant counterparts
who use one less strike than the limit.}

The results of these tests are presented in columns (2)\textendash (4)
and (6)\textendash (8) of Table \ref{tab:Defendant-Strike-Atty} for
all defendants and Black defendants, respectively. Across all columns,
we can see that including the controls for $n$ and $n-1$ case counts
and limiting the sample to public defenders does not weaken the results\textemdash if
anything, the coefficient of interest is slightly stronger, while
maintaining its statistical significance.\footnote{We also test various specifications using prosecutor and defense attorney
fixed effects (not reported here). Our coefficients remain positive
and similar in magnitude, but the large number of fixed effects (138
different defense attorneys and 94 different prosecutors across 305
cases) greatly increases standard errors, often making the results
insignificant.}

\subsection{Differences in crimes charged}

Another concern might be that defendants who use $n$ strikes are
charged with different crimes that have different conviction rates
than those who use $n-1$ strikes, even after controlling for broad
offense groups (e.g., misdemeanors). We can control for this possibility
through fixed effects based on more specific offense categories under
which the defendant was charged.\footnote{For a defendant who faces multiple charges, the classification is
based on the lowest count remaining to be decided as of the date of
jury selection (i.e., it was not dropped or otherwise adjudicated
before that date).} Following \citet{anwar2012impact}, we use the following offense
categories: Homicide, Other Violent Offenses, Property Offenses, Drug
Offenses, Sex Offenses, Weapons Offenses, and Other Offenses.\footnote{Our results remain largely robust if we use finer-grained offense
categories instead.}

Column (1) in Table \ref{tab:Defendant-Strike-Year-Charge-Judge-FEs}
shows ordinary least squares regression results for all defendants
after adding offense class fixed effects to the baseline regression;
column (5) shows the same regression for Black defendants. Our coefficient
of interest remains substantially the same and significant, suggesting
that differences in crimes charged are not driving our results.

\subsection{Heterogeneity in prosecutorial evidence}

Even after controlling for crimes charged, one might be concerned
that there is unobserved heterogeneity associated with the strength
of the case that is correlated with the number of strikes used. This
might be problematic if, for example, prosecutors typically have stronger
evidence against defendants who use $n$ strikes as compared to those
who use $n-1$ strikes. To address this concern, we can quantify evidentiary
strength by reviewing a case's probable cause affidavit, which lists
the pertinent facts and provides detailed information on the types
of evidence that law enforcement have collected at the time of arrest.
Specifically, we can code for whether prosecutors have the following
types of evidence that tie the defendant to the alleged crime: forensics,
recovered incriminating items, surveillance photos/videos, documentation,
a confession, or any admission of guilt by the defendant (including
ones that fall short of a confession). We can also control for whether
the crime occurred in a correctional facility or not.

We include all of these variables as controls in each specification
presented in Table \ref{tab:Defendant-Strike-Year-Charge-Judge-FEs}.
While each of these types of evidence is separately positively correlated
with guilty outcomes, none of them affects the magnitude or significance
of our coefficient of interest (defendant strike exhaustion) when
included as controls in our regressions. This provides greater confidence
that unobserved differences in case or evidentiary strength are not
driving our results.

\subsection{Court/judge behavior\label{subsec:last-story}}

Yet another potential concern is that judges might treat parties who
exhaust their strikes differently from those who do not. For example,
suppose a judge is annoyed with a party who uses more peremptory strikes,
because it causes voir dire to take more time, and hence the judge
is more likely to rule against that party on motions raised in the
subsequent litigation. This seems unlikely to be a concern at the
voir dire stage, when most pre-trial motions have already been decided,
and it seems even more unlikely when comparing cases in which $n-1$
strikes were used versus $n-$strike cases, since these two sets of
cases are likely very similar in terms of resources expended by the
court during jury selection.

Nonetheless, even if such an effect exists, we can arguably control
for it by using judge fixed effects, which allow us to pick up judge-specific
factors that might affect our result. Columns (2) and (5) in Table
\ref{tab:Defendant-Strike-Year-Charge-Judge-FEs} present our results
for all defendants and Black defendants, respectively. Once again,
we can see our coefficient of interest on peremptory strike exhaustion
remains substantially the same and is statistically significant.

\subsection{Multiple effects}

It is possible that more than one of the stories discussed in Sections
\ref{subsec:first-story} through \ref{subsec:last-story} are true,
and are working together to drive the results. While our limited sample
size does not permit us to test all combinations that might affect
the results, we can include all of the controls described above in
a single regression. We present these results in columns (3) and (6)
of Table \ref{tab:Defendant-Strike-Year-Charge-Judge-FEs}, for all
defendants and Black defendants, respectively. Again the coefficient
of interest maintains the same magnitude and level of statistical
significance, providing more suggestive evidence that our identification
strategy is valid.

\subsection{Other measures of strike usage}

We can also show that our results remain robust to different measures
of peremptory strike usage. For example, an alternate specification
might compare the class of defendants who have exhausted their peremptory
strikes (i.e., used $n$ strikes) with any defendant who has used
fewer strikes than the limit (i.e., $<n$ strikes). This change to
our baseline specification involves removing $n/n-1$ Group as a control,
and is shown in columns (1) and (5) of Table \ref{tab:Defendant-Strike-Diff-Defs}
for all defendants and Black defendants, respectively. As is apparent,
the change does not materially affect our results and if anything
increases the statistical significance.

As a further test, we can limit the sample to just those cases in
which the prosecution exhausted all of its $n$ strikes, or where
it used $n-1$ strikes. While this approach greatly reduces the sample
size, it also reduces the possibility of strategic interactions between
prosecution and defense affecting our results. Columns (2) and (6)
of Table \ref{tab:Defendant-Strike-Diff-Defs} show our results when
we limit the sample to cases in which the prosecution used up all
its strikes; columns (3) and (7) show the results when the prosecution
used $n-1$ strikes. Our results generally appear to be even stronger
when we limit the sample either way\textemdash in other words, the
effect of defendant strike exhaustion on conviction rates is even
more pronounced when the prosecution has used either $n-1$ or all
$n$ of its strikes.

Finally, we might wish to limit the sample to just Black and \white{}
defendants, rather than comparing Black to non-Black defendants. This
also does not significantly change results, as shown in column (4)
of Table \ref{tab:Defendant-Strike-Diff-Defs}, primarily because
the large majority of defendants in the sample is one of these two
races.

\subsection{Settlements and different outcome measures}

If parties can identify whether a particular jury is good or bad for
them, that might influence case outcomes even when the parties decide
to settle a case after jury selection. As such, in our primary specifications,
we look at all case outcomes, not just cases that proceed to jury
verdicts. Looking at all case outcomes rather than just those that
lead to verdicts also makes sense if we are concerned that parties
might vary in their risk preferences and hence might vary in their
willingness to settle after jury selection.\footnote{In our full dataset, there is a slight positive correlation between
defendants who exhaust their strikes and those who proceed to trial,
though the result is not statistically significant.} 

Nonetheless, we can also compare our results to the subset of cases
in which a jury issued a verdict. These specifications, as shown in
Table \ref{tab:Defendant-Strike-Diff-Guilt-Defs}, suggest our results
remain relatively robust across different measures of guilt, with
our coefficient of interest largely similar in magnitude to that in
our primary specifications. This is true both when we look at all
defendants (columns (1)\textendash (5)) and when we limit our sample
to Black defendants (columns (6)\textendash (10)).

Columns (1) and (6) report a baseline similar to our primary specification,
in which our outcome variable is the proportion of charges on which
a defendant was found guilty in a jury verdict.\footnote{\citet{doi:10.1086/698193} uses a similar specification as a robustness
check.} So if jury issues a guilty verdict on counts 1 and 3 but acquits
on count 2, then the defendant would be assigned an outcome of 0.667.
Again, our results are similar to our primary specifications and significant
at at least the 10\% level.

Columns (2) and (7) next report a specification in which a jury issued
a guilty or not guilty verdict on at least one count, with cases coded
as ``guilty'' if a defendant is found guilty on at least one count.
Here, the coefficient on defendant strike exhaustion decreases slightly
in magnitude and loses statistical significance. However, coding guilt
in this manner introduces significant measurement error that might
downward bias the coefficient, since 11.9\% of cases that go to a
verdict involve mixed verdicts in which a jury convicts on one count
but acquits on another. Categorizing all of these cases as defendant
losses is arguably inaccurate. Accordingly, the remaining columns
in the table categorize these mixed verdict cases in other ways.

Columns (3) and (8) code a case as guilty or not guilty based on the
jury's verdict on the lowest count it adjudicated, since the lowest
count generally corresponds to the most severe charge a defendant
faces. So if a jury issued a not guilty verdict on count 1 but a guilty
verdict on count 2, then the case would be coded as not guilty. When
coded this way, we can see our coefficients of interest increase and
become statistically significant at the 10\% level for all defendants
and also when we limit our sample to Black defendants.\footnote{Still, focusing on the lowest charge is likely inaccurate for many
cases as well. To illustrate, in one case in the dataset, a defendant
was found not guilty by a jury on count 1 (second degree murder with
a firearm) but guilty on count 2 (felon in possession of a firearm)
and sentenced to 15 years in prison. Classifying a case like this
as a defendant ``win'' is arguably misleading.}

Columns (4) and (9) deal with mixed verdict cases by removing them
altogether from the sample. The coefficient of interest in these specifications
is similar in terms of magnitude and statistical significance to our
primary specifications.

Columns (5) and (10) categorize mixed verdict cases as guilty if they
result in a jail/prison sentence for the defendant and not guilty
otherwise. Again, we find our coefficients are statistically significant
and similar to our primary specifications.

\subsection{Asymmetric information and risk aversion}

Another concern might be that asymmetric information or differences
in risk preferences somehow drive our results here. One issue might
be if a defendant knows less about the potential replacement juror
than the current juror under consideration. In that scenario, a more
risk-averse defendant might be less likely to exercise a peremptory
strike than a less risk-averse defendant. And if a defendant's risk
aversion during jury selection is inversely related to the strength
of the prosecutor's case\textemdash which might occur if a defendant
who has a strong case feels less obliged to take risks during jury
selection\textemdash then we might expect defendants who exhaust their
peremptory strikes to have worse quality cases (and hence worse outcomes)
than those who do not exhaust them.

Although such a concern might be relevant in other settings, it likely
does not affect our results here. This is because in Florida, parties
have the same information on all potential jurors in the pool, since
all pool members are questioned at the outset by the judge and both
parties.\footnote{There are apparently relatively few cases in which a court runs out
of jury pool members, since jury administrators take into account
the type of case scheduled for trial and a judge's preferences when
deciding how many people to summon and how large the initial pool
for a case should be. On the rare occasions when this occurs, information
about potential jurors would be revealed in two separate stages of
voir dire.} 

A variant of this concern might be that a defendant only knows the
average predisposition for each juror in the pool rather than the
full distribution of potential outcomes for each juror. To illustrate,
suppose there are two otherwise identical Black defendants (defendant
1 and defendant 2), both charged with the same crime and facing the
same potential prison sentence. Defendant 1 knows he is innocent and
the evidence is on his side, so he expects there is a small chance
of being found guilty. As such, he does not need to take chances during
jury selection and can play it relatively safe. By contrast, defendant
2 knows she is guilty and the evidence is against her, so she expects
a higher chance of being found guilty. This defendant's only chance
of receiving a \textquotedbl not guilty\textquotedbl{} verdict or
a hung jury is to take a chance on the jurors.

Suppose the current juror under consideration for either defendant
is a \white{} man, and the replacement juror is a \white{} woman.
Suppose also that all \white{} men are identical in terms of their
predisposition toward Black defendants, but there are two types of
\white{} women\textemdash those who are are very likely to convict
a Black defendant and those who are very unlikely to convict a Black
defendant\textemdash and that attorneys cannot distinguish between
these two types. Even if, on average, a \white{} woman is more likely
to convict the defendant than a \white{} male, defendant 2 might
take a chance on the \white{} woman since that is his only chance
to win the case. As such, defendant 1 and defendant 2 might diverge
in terms of their striking behavior.

While such a concern might also be relevant, it is unlikely to affect
our results because there is no reason to believe a higher variance
juror (i.e., the \white{} woman in the example above) will be more
likely to appear in a replacement slot versus the current juror slot.
This is because potential jurors are randomly ordered in the pool.
So on average, risk aversion should not bias our results, even though
it might affect how a particular defendant uses his strikes. Defendant
2 will end up with a riskier juror, but there is no reason to believe
he is more likely to end up in the $n$ strike class versus the $n-1$
strike class.\footnote{Another potential concern might be that the variance of juror predispositions
increases as the quality of a defendant's case worsens. Put differently,
jurors on the whole might become riskier when defendants have a poor
quality case. If so, we might expect defendants to use more strikes
when they have worse cases, since they are facing a more extreme jury
pool. If this theory were true, we should expect to see a steady increase
in conviction rates as defendants use additional strikes, rather than
a large jump at the strike exhaustion boundary. But as discussed in
the next section, this is not what we see when we conduct placebo
tests.}

\subsection{Placebo tests}

Finally, we can conduct placebo tests to test whether there is anything
special about cases in which parties exhaust peremptory strikes. In
particular, instead of comparing the difference between $n$ and $n-1$
strike cases, we can imagine the strike boundary to be at $n-1$ instead,
and hence compare $n-1$ and $n-2$ strike cases. If strike exhaustion
really matters, as we posit it does, then we should see a much diminished
effect\footnote{One might still expect a jury pool to be more favorable to the striking
party in an $n-2$ strike case than in an $n-1$ strike case because
exercising the penultimate strike is still costly, as it removes the
possibility of striking two prospective jurors later on in voir dire.
Nonetheless, the cost of the penultimate strike should be less than
the cost of the last strike if the only difference on average between
different strike groups is differences in jury composition.} or no effect at all at the placebo boundary.

Table \ref{tab:Defendant-Placebo-Strike-Exhaustion} shows the same
specifications as in columns (1)\textendash (3) of Tables \ref{tab:Defendant-Pl-Strike-Exhaustion}
and \ref{tab:Defendant-Bl-NBl-Exhaustion}, except now defendant strike
exhaustion is defined at the placebo boundary. The coefficient on
placebo strike exhaustion is insignificant and close to zero in all
specifications, whether we look across all cases or just cases with
Black defendants. These results provide additional support that differences
in jury composition, and not just another variable correlated with
an increased use of peremptory strikes, is what drives our primary
results.

\section{Conclusion\label{sec:Policy and Conclusion}}

In this paper, we provide new evidence on how random variation in\emph{
}jury composition affects criminal case outcomes. Our approach differs
from previous related empirical work, which focused only on the effect
of individual variables observable at the time the initial jury pool
was created.

The central insight is that the selection process from pool to box
can be exploited to extract exogenous information on the jury pool
biases. Specifically, strike exhaustion act as signals of how litigants
view the jury pool they have been randomly assigned. All else being
equal, we expect a litigant who can accurately assess whether a juror
favors her side would use more strikes when she has been assigned
an unfavorable pool than when she has been assigned a favorable one.
Thus, litigants' use of strikes enables us to identify juror predisposition,
which is otherwise unmeasurable. The identifying assumption in the
most conservative specification is that cases in which litigants use
all of their peremptory strikes ($n$ strike group) are on average
identical to cases in which litigants use one less strike than the
limit ($n-1$ strike group), except for differences in the jury pool
they have been assigned.

Using recent data from a large, racially diverse county in Florida,
we find our $n$ and $n-1$ strike groups appear largely similar on
a broad range of observable pre-trial characteristics, including defendant
demographics and past imprisonment history, attorney experience and
education, charged offenses, and type of prosecutorial evidence. We
find defendants who use all $n$ of their peremptory strikes are significantly
more likely to be convicted than defendants who use one less strike
than the limit. This result is driven by Black defendants, for whom
strike exhaustion increases conviction rates by 18\textendash 20 percentage
points. We explain why these results are likely a lower bound on the
effect of jury composition on case outcomes. We also run multiple
robustness checks with a battery of covariates and across various
subsamples; our coefficient of interest remains remarkably stable
and significant across nearly every specification.

A primary contribution of this paper is a new framework for identification,
which enables to show that jury composition greatly affects case outcomes,
particularly for Black defendants. In addition, we believe our paper
makes a number of other contributions.

First, our results are the first empirical evidence of the causal
impact of peremptory strikes on case outcomes. Scholars have questioned
for decades whether peremptory strikes affect how a case turns out,
or whether attorneys can correctly identify during the strike process
which jurors are likely to be predisposed toward their side.\footnote{See, e.g., \citet{ford2009modeling}, \citet{zeisel1977effect}.}
Our research suggests that peremptory strikes (and the limits placed
on those strikes) do in fact affect case outcomes, and that attorneys
are at least somewhat effective in using those strikes to shape the
final jury.\footnote{Even though the vast majority of criminal cases result in plea bargains
that occur before jury selection, a long literature discusses how
this bargaining process might occur in the shadow of the expected
results of the jury trial. Compare \citet{mnookin1978bargaining},
\citet{easterbrook1983criminal}, \citet{offit2018prosecuting} with
\citet{bibas2004plea}. If this is true, the phenomenon we identify
here might have a significant impact on pretrial negotiations, at
least in cases the parties were considering taking to trial. For example,
if Black defendants are risk-averse and they face more risk in the
jury lottery than non-Black defendants, they might be more likely
to settle cases before trial and on worse terms (because of weaker
bargaining power) than they would in the absence of this heightened
risk.} Our results also suggest increasing peremptory strike limits for
defendants would be one way to decrease the variance in outcomes for
similarly-situated Black defendants.\footnote{The federal court system and a number of state courts already give
defendants more peremptory strikes than prosecutors. See, e.g., Fed.
R. Crim. P. 24(b)(2) (defense and prosecution get 10 and 6 strikes,
respectively, in non-capital felonies); Maryland Rule 4-313(a)(2)
(defense and prosecution get 20 and 10 strikes, respectively, for
life-eligible crimes); Minn. Crim. P. R. 26.02(9) (defense and prosecution
get 15 and 9 strikes, respectively, for life-eligible crimes); New
Mexico Crim. P. R. 5-606(D)(b) (defense and prosecution get 12 and
8 strikes, respectively, for non-capital, life-eligible crimes).}

To be sure, our paper cannot settle larger debates among scholars
whether peremptory challenges are socially beneficial or harmful,\footnote{See, e.g., \citet{flanagan2015peremptory}, \citet{babcock1974voir},
\citet{howard2010taking}.} or whether they make it more likely that guilty people are convicted
and innocent people are acquitted in jury trials. Answering the latter
question in particular would require, among other things, knowing
whether defendants actually committed the underlying crimes for which
they are charged, which we cannot determine. Regardless, the wide
variation in outcomes we measure, driven by pure chance and concentrated
among Black defendants, raises concerns. This problem is magnified
given that a large percentage of defendants in our data exhausted
their strikes and thus might be affected by the phenomenon we identify
here.

Relatedly, the approach we describe here might be a useful diagnostic
tool to determine whether jury selection rules in a particular jurisdiction
are problematic. Put differently, if parties are awarded a sufficient
number of peremptory strikes, either the proportion of cases with
$n$ and $n-1$ strikes should be small, or we should not find any
statistically significant difference in conviction rates between those
groups. The presence of large differences might be evidence that random
variation in jury composition is having an outsized effect on criminal
case outcomes in that jurisdiction.

Finally, our paper raises deeper policy questions about how jury selection
should be conducted so that the variance (across cases) of predispositions
in the final jury is small while preserving the random assignment
of jury pools. Several factors that affect jury composition need to
be further isolated and examined. For example, should attorneys know
who the next potential juror is when they decide whether to use a
strike (as is the case in Florida but not in some other states)? Should
attorneys be allowed to view the jury pool when making strike decisions,
so as to see the age, race, and sex of potential jurors, even though
the latter two are protected classes on which attorneys are not permitted
to discriminate? How should the number of strikes be chosen and how
should it relate to the final jury size? Should the usage of Batson
challenges be extended to take into account the juror \emph{replacing
}the one being struck? Our paper motivates why answers to these questions
could be important for the design of a fair voir dire process. As
such, we hope our work sets the stage for further theoretical and
empirical research on jury selection.

We have explored and identified an important determinant of justice.
Fairness is multidimensional. Not only should the judicial system
be designed for impartiality \emph{on average }(across cases), a fair
system should have low variance so that defendants can be confident
their case is likely to receive a similar outcome to cases that are
similar. By exposing this dimension of the justice system, we hope
to spur future research in law and in mechanism design to achieve
an optimized level of justice.

\newpage{}

\section*{Figures}

\addcontentsline{toc}{section}{Figures}

\begin{figure}
[ph!]
\caption{Peremptory Strikes: Defense\label{fig:Peremptory-Strikes:-Defendant}}

\noindent\noindent\resizebox{.99\textwidth}{!}{%

\includegraphics{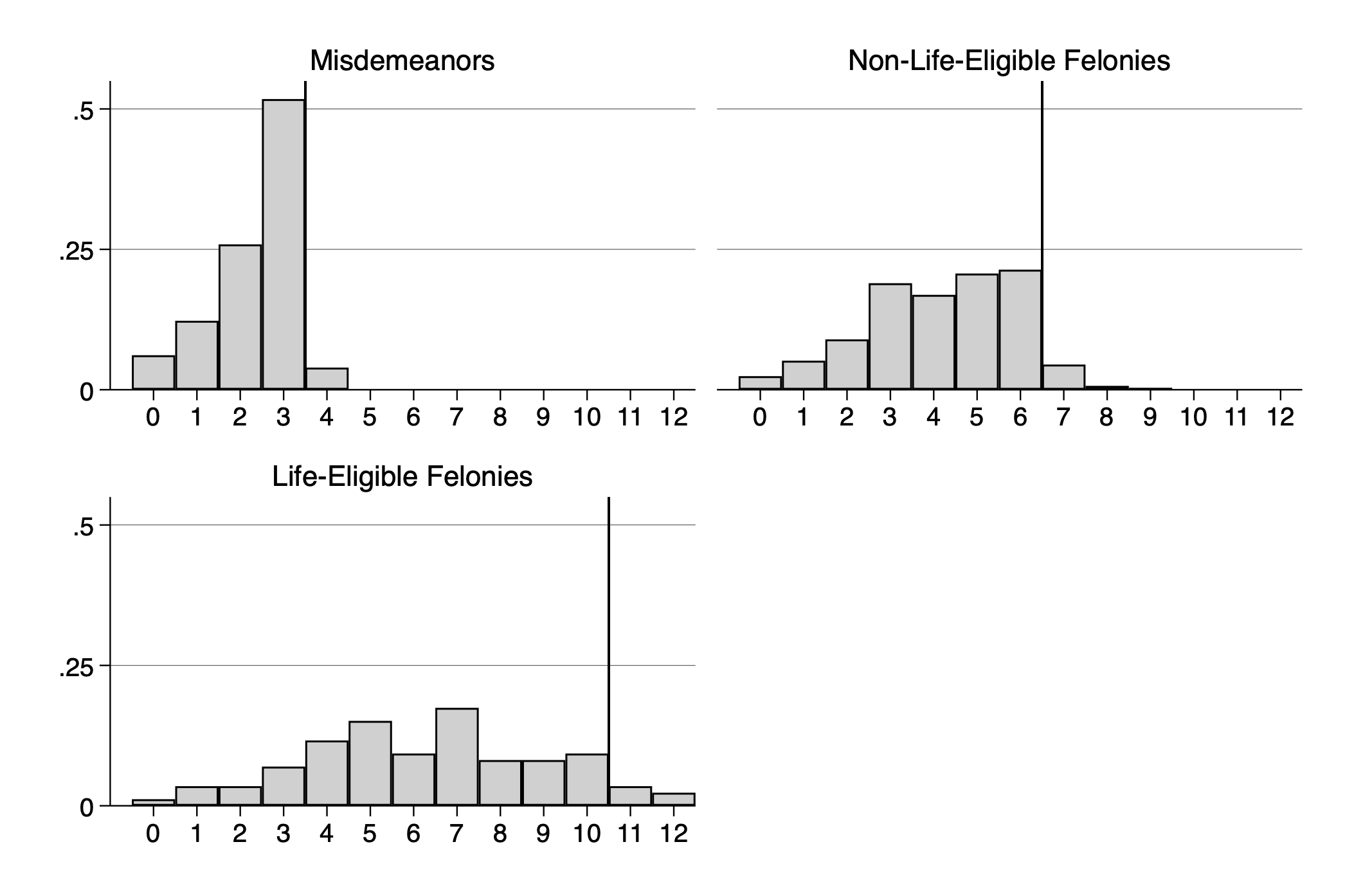}

}

\begin{threeparttable}
\begin{tablenotes} 
\item\emph{Notes}: This figure shows histograms for the number of peremptory strikes used by the defense in the 604 trials from January 2015 through September 2017 in our full dataset (includes cases in which parties exceeded their peremptory strike limits). The peremptory strike limits for the three offense classes of misdemeanor, non-life-eligible felony, and life-eligible felony are 3, 6, and 10, respectively, as shown by the vertical lines.
\end{tablenotes} 
\end{threeparttable}
\end{figure}

\begin{figure}
[ph!]
\caption{Peremptory Strikes: Prosecution\label{fig:Peremptory-Strikes:-Prosecution}}

\noindent\noindent\resizebox{.99\textwidth}{!}{%

\includegraphics{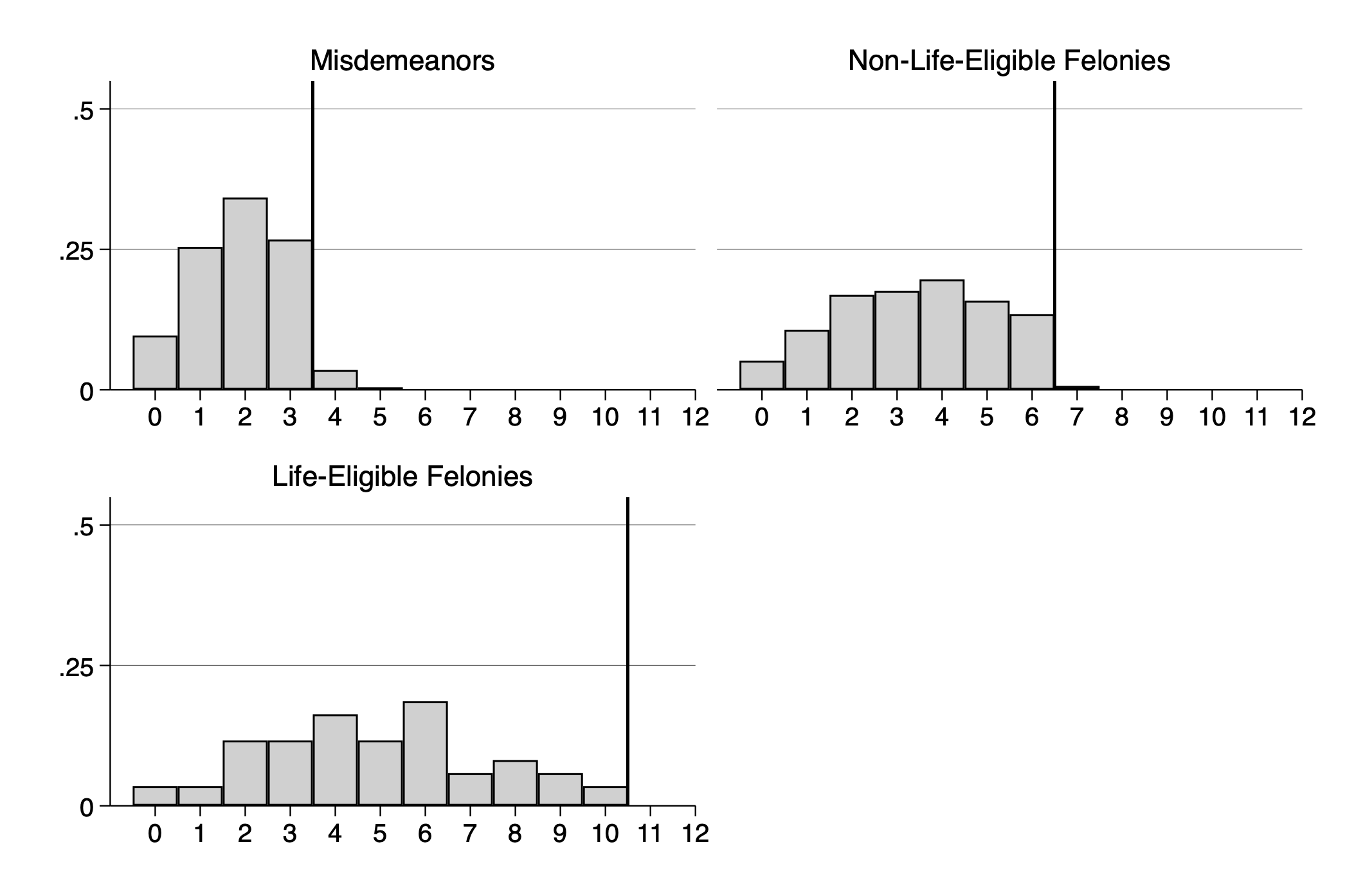}

}

\begin{threeparttable}
\begin{tablenotes} 
\item\emph{Notes}: This figure shows histograms for the number of peremptory strikes used by the prosecution in the 604 trials from January 2015 through September 2017 in our full dataset (includes cases in which parties exceeded their peremptory strike limits). The peremptory strike limits for the three offense classes of misdemeanor, non-life-eligible felony, and life-eligible felony are 3, 6, and 10, respectively, as shown by the vertical lines.
\end{tablenotes} 
\end{threeparttable}
\end{figure}

\begin{figure}
[ph!]
\caption{Proportion of Crimes Across Peremptory Strike Classes: Defense\label{fig:Def-Crime-Peremptory-Strikes}}

\noindent\noindent\resizebox{.99\textwidth}{!}{%

\includegraphics{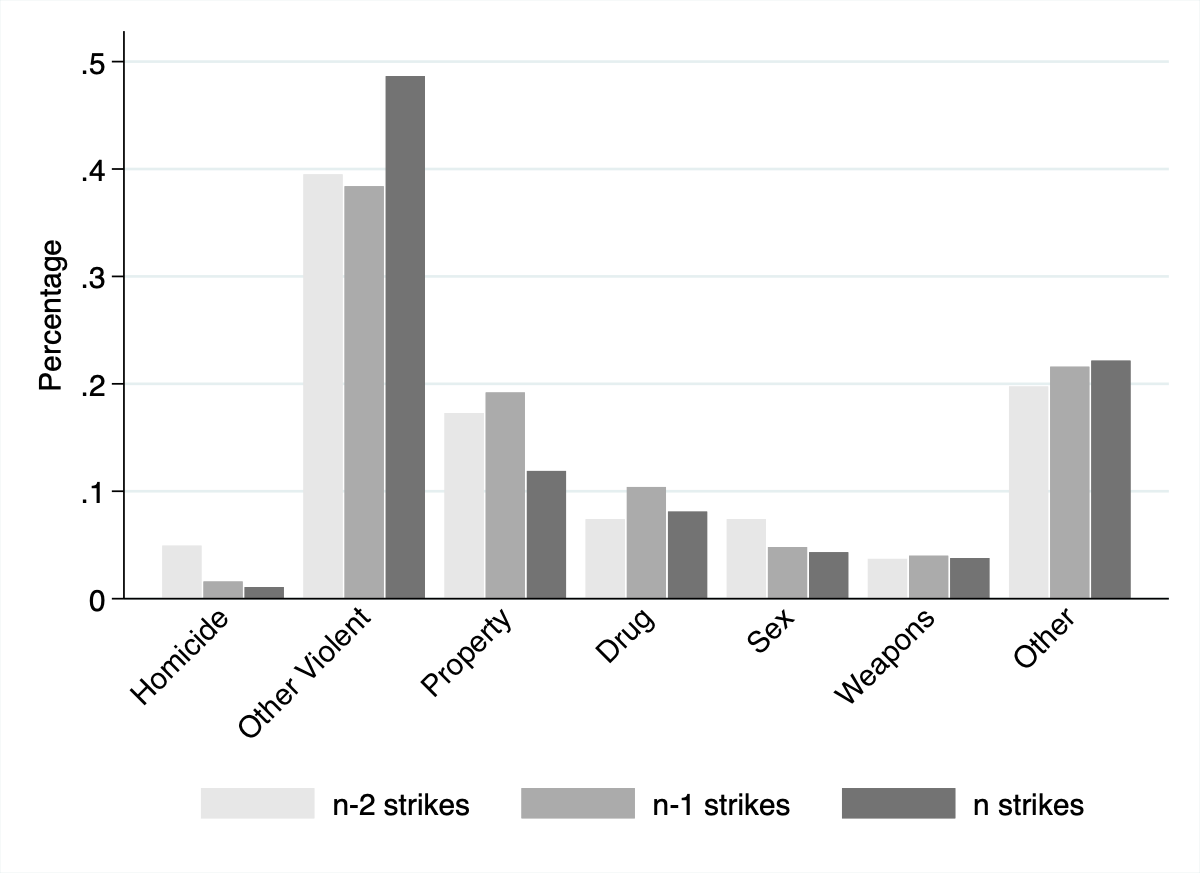}

}

\begin{threeparttable}
\begin{tablenotes} 
\item\emph{Notes}: This figure shows the percentage of cases in each of seven different offense categories (as used in Anwar et al. 2012), across three different peremptory strike groups for defendants. The $n$, $n-1$ and $n-2$ strike groups comprise defendants who used all of their peremptory strikes, one less strike than the limit, and two less strikes than the limit, respectively.
\end{tablenotes} 
\end{threeparttable}
\end{figure}

\begin{figure}
[ph!]
\caption{Proportion of Crimes Across Peremptory Strike Classes: Prosecution\label{fig:Pl-Crime-Peremptory-Strikes-1}}

\noindent\noindent\resizebox{.99\textwidth}{!}{%

\includegraphics{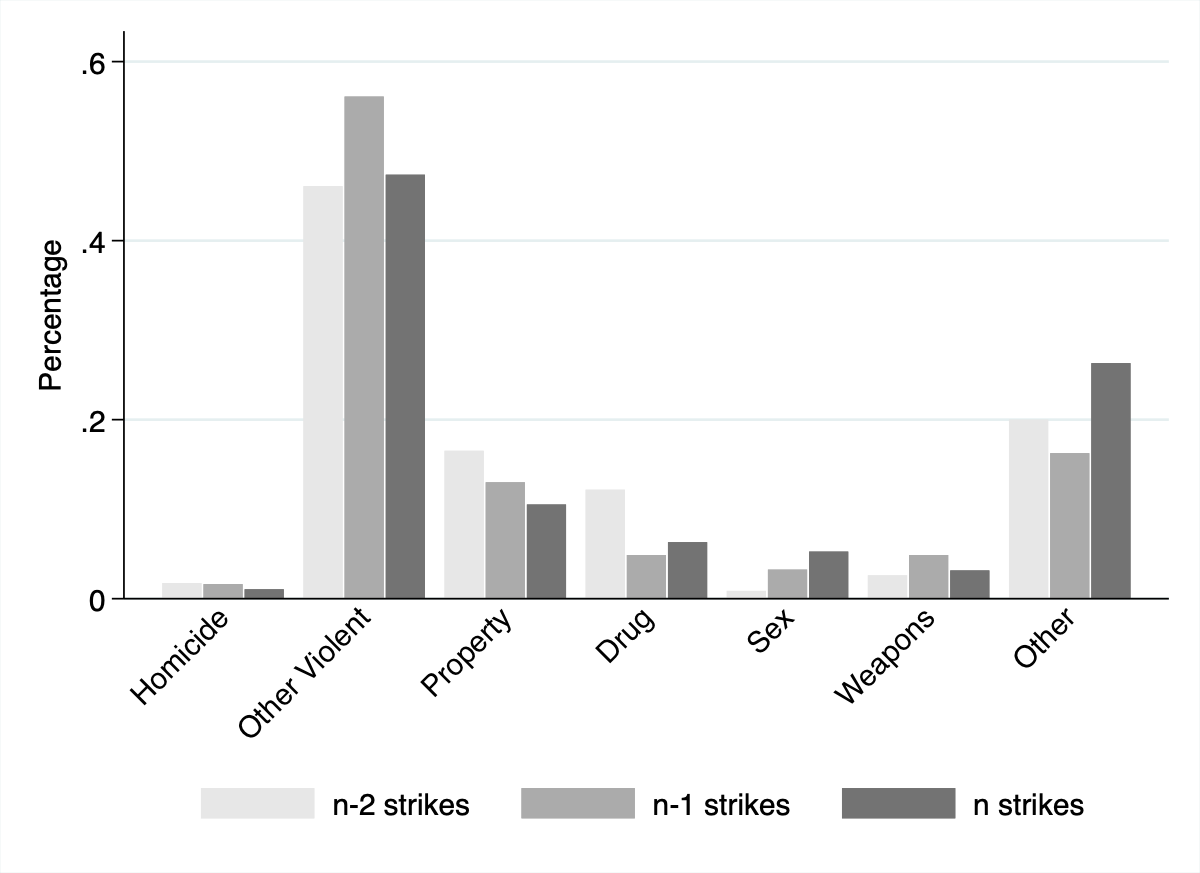}

}

\begin{threeparttable}
\begin{tablenotes} 
\item\emph{Notes}: This figure shows the percentage of cases in each of six different offense categories (as used in Anwar et al. 2012), across three different peremptory strike groups for the prosecution. The $n$, $n-1$ and $n-2$ strike groups comprise prosecutors who used all of their peremptory strikes, one less strike than the limit, and two less strikes than the limit, respectively.
\end{tablenotes} 
\end{threeparttable}
\end{figure}

\section*{\newpage Tables }

\addcontentsline{toc}{section}{Tables}

\begin{table}
[ph!]
\begin{center}
\begin{threeparttable}\caption{Hypothetical Jury Selection Strike Sheet\label{tab:Voir-Dire-Strike-Sheet}}

\begin{tabular}{llrrrrrrrrrr}
\multicolumn{12}{c}{Seat \#}\tabularnewline
\hline 
 &
 &
1 &
2 &
3 &
4 &
5 &
6 &
7 &
8 &
9 &
10\tabularnewline
\hline 
\hline 
\multirow{3}{*}{Row \#} &
1 &
JC &
JC &
J-1 &
D-1 &
D-2 &
J-2 &
JC &
JC &
P-1 &
J-3\tabularnewline
 & 2 &
DC &
J-4 &
JC &
D-3 &
PC &
D-4 &
P-2 &
D-5 &
P-3 &
JC\tabularnewline
 & 3 &
J-5 &
D-6 &
P-4 &
DC &
J-6 &
NU &
PC &
NU &
JC &
NU\tabularnewline
\hline 
\end{tabular}

\begin{tablenotes} 
\item\emph{Notes}: This table shows a hypothetical jury selection strike sheet for a non-life-eligible felony case. The pool comprises 30 jurors, ordered from row 1, seat 1, to row 3, seat 10. JC: pool members struck for cause by the judge; DC and PC: pool members struck for cause by the defense and prosecution, respectively. D-\# and P-\#: pool members for whom the defense and prosecution, respectively, issued peremptory strikes, with the \# of the strike used. J-\#: pool members selected for final jury. NU: pool members who were neither struck nor used in the final jury.
\end{tablenotes} 
\end{threeparttable}
\end{center}
\end{table}

\newpage{}
\begin{table}[H]
\begin{sideways}
\noindent\resizebox{.70\textheight}{!}{%
\begin{threeparttable}\caption{Defendant, Attorney and Pool Characteristics in $n$, $n-1$, and
$n-2$ Strike Groups by Party\label{tab:Summary-Statistics}}

\begin{tabular}{lr||rrrr||rrrr}
\hline 
 &
\multicolumn{1}{r}{} &
\multicolumn{8}{c}{}\tabularnewline
 &
 &
\multicolumn{4}{c||}{Defense} &
\multicolumn{4}{c}{Prosecution}\tabularnewline
\cline{2-10}
 &
All &
$n$ strikes &
$n-1$ strikes &
$n-2$ strikes &
$n$\emph{ v.}\ $n-1$ &
$n$ strikes &
$n-1$ strikes &
$n-2$ strikes &
$n$\emph{ v.}\ $n-1$\tabularnewline
 &
(1) &
(2) &
(3) &
(4) &
(5) &
(6) &
(7) &
(8) &
(9)\tabularnewline
\hline 
\hline 
\uline{Defendant characteristics} &
 &
 &
 &
 &
 &
 &
 &
 &
\tabularnewline
Age &
36.84 &
37.52 &
35.39 &
38.32 &
2.13 &
36.43 &
38.12 &
37.53 &
-1.69\tabularnewline
 &
(12.68) &
(13.27) &
(11.76) &
(12.52) &
(1.47) &
(11.69) &
(13.77) &
(12.58) &
(-0.97)\tabularnewline
Black &
0.55 &
0.52 &
0.52 &
0.49 &
-0.00 &
0.55 &
0.47 &
0.53 &
0.08\tabularnewline
 &
(0.50) &
(0.50) &
(0.50) &
(0.50) &
(-0.02) &
(0.50) &
(0.50) &
(0.50) &
(1.11)\tabularnewline
Hispanic &
0.13 &
0.14 &
0.14 &
0.14 &
0.00 &
0.17 &
0.13 &
0.12 &
0.04\tabularnewline
 &
(0.34) &
(0.35) &
(0.34) &
(0.34) &
(0.11) &
(0.38) &
(0.34) &
(0.33) &
(0.78)\tabularnewline
Female &
0.14 &
0.14 &
0.14 &
0.17 &
-0.00 &
0.13 &
0.10 &
0.17 &
0.03\tabularnewline
 &
(0.34) &
(0.35) &
(0.35) &
(0.38) &
(-0.09) &
(0.33) &
(0.30) &
(0.38) &
(0.66)\tabularnewline
\# Prev.\ Imprisonments &
0.55 &
0.43 &
0.52 &
0.74 &
-0.09 &
0.55 &
0.46 &
0.51 &
0.08\tabularnewline
 &
(1.32) &
(1.05) &
(1.31) &
(1.97) &
(-0.66) &
(1.30) &
(1.37) &
(1.33) &
(0.46)\tabularnewline
Years Prev.\ Imprison. &
1.26 &
1.33 &
0.73 &
1.20 &
0.60{*} &
1.37 &
1.20 &
1.00 &
0.18\tabularnewline
 &
(3.51) &
(4.12) &
(1.90) &
(3.46) &
(1.72) &
(3.88) &
(3.78) &
(2.47) &
(0.34)\tabularnewline
\uline{Attorney characteristics} &
 &
 &
 &
 &
 &
 &
 &
 &
\tabularnewline
Defense Experience &
9.71 &
7.82 &
9.66 &
10.31 &
-1.84 &
7.65 &
8.17 &
8.82 &
-0.52\tabularnewline
 &
(10.32) &
(10.02) &
(11.01) &
(10.50) &
(-1.48) &
(10.18) &
(10.73) &
(10.26) &
(-0.36)\tabularnewline
Prosecutor Experience &
5.51 &
3.63 &
5.18 &
5.88 &
-1.55{*}{*} &
3.71 &
4.27 &
4.97 &
-0.56\tabularnewline
 &
(5.66) &
(4.33) &
(6.38) &
(4.82) &
(-2.33) &
(4.32) &
(4.72) &
(6.18) &
(-0.90)\tabularnewline
Def.\ Law Sch.\ Rank &
87.86 &
80.99 &
95.02 &
93.27 &
-14.03{*}{*} &
92.28 &
80.83 &
89.83 &
11.46\tabularnewline
 &
(54.84) &
(53.70) &
(57.15) &
(56.66) &
(-2.14) &
(53.80) &
(53.43) &
(54.66) &
(1.54)\tabularnewline
Pr.\ Law Sch.\ Rank &
100.33 &
96.90 &
96.48 &
97.62 &
0.42 &
95.55 &
101.37 &
94.76 &
-5.82\tabularnewline
 &
(56.05) &
(54.87) &
(55.87) &
(57.03) &
(0.06) &
(57.66) &
(55.69) &
(56.26) &
(-0.74)\tabularnewline
Public Defender &
0.75 &
0.77 &
0.72 &
0.73 &
0.05 &
0.81 &
0.77 &
0.76 &
0.04\tabularnewline
 &
(0.44) &
(0.42) &
(0.45) &
(0.45) &
(0.93) &
(0.39) &
(0.42) &
(0.43) &
(0.69)\tabularnewline
\uline{Pool characteristics} &
 &
 &
 &
 &
 &
 &
 &
 &
\tabularnewline
Prop.\ Black in Pool &
0.13 &
0.13 &
0.13 &
0.14 &
-0.00 &
0.12 &
0.13 &
0.13 &
-0.01\tabularnewline
 &
(0.06) &
(0.06) &
(0.07) &
(0.06) &
(-0.19) &
(0.08) &
(0.06) &
(0.07) &
(-1.23)\tabularnewline
Prop.\ Female in Pool &
0.54 &
0.53 &
0.53 &
0.54 &
0.00 &
0.53 &
0.53 &
0.55 &
0.00\tabularnewline
 &
(0.09) &
(0.09) &
(0.10) &
(0.09) &
(0.02) &
(0.09) &
(0.09) &
(0.10) &
(0.23)\tabularnewline
Avg.\ Age in Pool &
46.91 &
47.12 &
46.70 &
46.53 &
0.42 &
47.05 &
46.79 &
46.82 &
0.26\tabularnewline
 &
(3.19) &
(3.31) &
(3.60) &
(2.69) &
(1.03) &
(3.24) &
(3.19) &
(3.64) &
(0.59)\tabularnewline
Median Pool Income &
62,404 &
62,596 &
62,052 &
62,366 &
544 &
63,205 &
62,264 &
61,639 &
941\tabularnewline
 &
(4,976) &
(5,436) &
(5,017) &
(4,891) &
(0.91) &
(4,938) &
(5,352) &
(5,336) &
(1.34)\tabularnewline
Prop.\ Democrat in Pool &
0.29 &
0.30 &
0.30 &
0.29 &
-0.00 &
0.29 &
0.30 &
0.30 &
-0.01\tabularnewline
 &
(0.08) &
(0.08) &
(0.08) &
(0.08) &
(-0.34) &
(0.07) &
(0.09) &
(0.09) &
(-0.70)\tabularnewline
Prop.\ Republican in Pool &
0.21 &
0.21 &
0.21 &
0.21 &
0.00 &
0.23 &
0.21 &
0.20 &
0.02{*}\tabularnewline
 &
(0.08) &
(0.08) &
(0.08) &
(0.07) &
(0.19) &
(0.08) &
(0.09) &
(0.08) &
(1.70)\tabularnewline
\hline 
Observations &
567 &
185 &
125 &
81 &
\textendash{} &
95 &
123 &
115 &
\textendash{}\tabularnewline
\hline 
\end{tabular}

\begin{tablenotes} 
\item\emph{Notes}: This table compares baseline defendant, attorney, and jury pool characteristics across different peremptory strike groups. Column (1) shows mean values across all cases in our primary dataset. The $n$, $n-1$, and $n-2$ strike groups comprise parties who exhausted all of their peremptory strikes, one less strike than the limit, and two less strikes than the limit, respectively, for defendants (columns (2)--(4)) and prosecutors (columns (6)--(8)). Attorney experience is number of years between attorney admittance to the Florida state bar and date of jury selection. Law school rankings are calculated based on 2018 U.S. News \& World Report rankings. Black, Hispanic, Female and Public Defender are dummy variables for these respective characteristics. Number of previous imprisonments is the number of previous stints a defendant had in Florida state prison; years of previous imprisonment is the total prison term for those stints. Proportion Black and proportion female are the proportion of Black and female jurors, respectively, in the pre-strike jury pool. Average age is the average age of pool members. Median pool income is calculated by estimating each juror's income by the median income in the zip code in which she resides (using U.S. Census data and 2017 inflation-adjusted dollars). Values in parentheses are standard deviations for columns (1)--(4) and (6)--(8), and t-values for two-sided t-tests in columns (5) and (9). ** = significant at 5\% level, * = significant at 10\% level.
\end{tablenotes} 
\end{threeparttable}
}
\end{sideways}
\end{table}

\newpage{}
\begin{table}[H]
\begin{sideways}
\noindent\resizebox{.8\textheight}{!}{%
\begin{threeparttable}\caption{Case Characteristics in $n$, $n-1$, and $n-2$ Strike Groups by
Party\label{tab:Summary-Statistics-1}}

\begin{tabular}{lr||rrrr||rrrr}
\hline 
 &
\multicolumn{1}{r}{} &
\multicolumn{8}{c}{}\tabularnewline
 &
 &
\multicolumn{4}{c||}{Defense} &
\multicolumn{4}{c}{Prosecution}\tabularnewline
\cline{2-10}
 &
All &
$n$ strikes &
$n-1$ strikes &
$n-2$ strikes &
$n$\emph{ v.}\ $n-1$ &
$n$ strikes &
$n-1$ strikes &
$n-2$ strikes &
$n$\emph{ v.}\ $n-1$\tabularnewline
 &
(1) &
(2) &
(3) &
(4) &
(5) &
(6) &
(7) &
(8) &
(9)\tabularnewline
\hline 
Homicide &
0.03 &
0.01 &
0.02 &
0.05 &
-0.01 &
0.01 &
0.02 &
0.02 &
-0.01\tabularnewline
 &
(0.18) &
(0.10) &
(0.13) &
(0.22) &
(-0.38) &
(0.10) &
(0.13) &
(0.13) &
(-0.37)\tabularnewline
Other Violent Offense &
0.40 &
0.49 &
0.38 &
0.40 &
0.10{*} &
0.47 &
0.56 &
0.46 &
-0.09\tabularnewline
 &
(0.49) &
(0.50) &
(0.49) &
(0.49) &
(1.79) &
(0.50) &
(0.50) &
(0.50) &
(-1.28)\tabularnewline
Property Offense &
0.18 &
0.12 &
0.19 &
0.17 &
-0.07{*} &
0.11 &
0.13 &
0.17 &
-0.02\tabularnewline
 &
(0.38) &
(0.32) &
(0.40) &
(0.38) &
(-1.71) &
(0.31) &
(0.34) &
(0.37) &
(-0.56)\tabularnewline
Drug Offense &
0.09 &
0.08 &
0.10 &
0.07 &
-0.02 &
0.06 &
0.05 &
0.12 &
0.01\tabularnewline
 &
(0.28) &
(0.27) &
(0.31) &
(0.26) &
(-0.67) &
(0.24) &
(0.22) &
(0.33) &
(0.45)\tabularnewline
Sex Offense &
0.07 &
0.04 &
0.05 &
0.07 &
-0.00 &
0.05 &
0.03 &
0.01 &
0.02\tabularnewline
 &
(0.26) &
(0.20) &
(0.21) &
(0.26) &
(-0.20) &
(0.22) &
(0.18) &
(0.09) &
(0.72)\tabularnewline
Weapons Offense &
0.04 &
0.04 &
0.04 &
0.04 &
-0.00 &
0.03 &
0.05 &
0.03 &
-0.02\tabularnewline
 &
(0.21) &
(0.19) &
(0.20) &
(0.19) &
(-0.10) &
(0.18) &
(0.22) &
(0.16) &
(-0.65)\tabularnewline
Other Offense &
0.18 &
0.22 &
0.22 &
0.20 &
0.01 &
0.26 &
0.16 &
0.20 &
0.10{*}\tabularnewline
 &
(0.39) &
(0.42) &
(0.41) &
(0.40) &
(0.12) &
(0.44) &
(0.37) &
(0.40) &
(1.78)\tabularnewline
Counts Charged &
2.44 &
2.10 &
2.43 &
3.00 &
-0.33 &
2.27 &
2.04 &
2.51 &
0.23\tabularnewline
 &
(2.37) &
(1.87) &
(2.69) &
(3.60) &
(-1.21) &
(2.68) &
(1.46) &
(2.47) &
(0.77)\tabularnewline
Correctional Setting &
0.02 &
0.02 &
0.01 &
0.02 &
0.01 &
0.00 &
0.03 &
0.01 &
-0.03{*}{*}\tabularnewline
 &
(0.13) &
(0.13) &
(0.09) &
(0.16) &
(0.69) &
(0.00) &
(0.18) &
(0.09) &
(-2.03)\tabularnewline
Photos/Video &
0.14 &
0.13 &
0.15 &
0.11 &
-0.02 &
0.11 &
0.13 &
0.17 &
-0.02\tabularnewline
 &
(0.35) &
(0.33) &
(0.36) &
(0.32) &
(-0.59) &
(0.31) &
(0.34) &
(0.37) &
(-0.55)\tabularnewline
Recovered Items &
0.28 &
0.23 &
0.27 &
0.36 &
-0.04 &
0.23 &
0.20 &
0.34 &
0.03\tabularnewline
 &
(0.45) &
(0.42) &
(0.45) &
(0.48) &
(-0.76) &
(0.42) &
(0.40) &
(0.48) &
(0.48)\tabularnewline
Forensic Evidence &
0.09 &
0.03 &
0.08 &
0.12 &
-0.05{*} &
0.05 &
0.03 &
0.06 &
0.02\tabularnewline
 &
(0.29) &
(0.18) &
(0.27) &
(0.33) &
(-1.68) &
(0.23) &
(0.18) &
(0.24) &
(0.72)\tabularnewline
Documentation &
0.22 &
0.15 &
0.18 &
0.28 &
-0.03 &
0.09 &
0.16 &
0.17 &
-0.07\tabularnewline
 &
(0.41) &
(0.36) &
(0.38) &
(0.45) &
(-0.63) &
(0.28) &
(0.37) &
(0.38) &
(-1.60)\tabularnewline
Confession &
0.15 &
0.17 &
0.14 &
0.15 &
0.02 &
0.15 &
0.17 &
0.10 &
-0.02\tabularnewline
 &
(0.36) &
(0.37) &
(0.35) &
(0.36) &
(0.54) &
(0.36) &
(0.38) &
(0.30) &
(-0.45)\tabularnewline
Admission &
0.40 &
0.37 &
0.38 &
0.49 &
-0.01 &
0.39 &
0.37 &
0.43 &
0.02\tabularnewline
 &
(0.49) &
(0.48) &
(0.49) &
(0.50) &
(-0.20) &
(0.49) &
(0.49) &
(0.50) &
(0.23)\tabularnewline
\hline 
Observations &
567 &
185 &
125 &
81 &
\textendash{} &
95 &
123 &
115 &
\textendash{}\tabularnewline
\hline 
\end{tabular}

\begin{tablenotes} 
\item\emph{Notes}: This table compares baseline case characteristics across different peremptory strike groups. Column (1) shows mean values across all cases in our primary dataset. The $n$, $n-1$, and $n-2$ strike groups comprise parties who exhausted all of their peremptory strikes, one less strike than the limit, and two less strikes than the limit, respectively, for defendants (columns (2)--(4)) and prosecutors (columns (6)--(8)). Correctional Setting is a dummy variable for whether the alleged offense took place in a prison/jail. Photos/Video, Recovered Items, Forensic Evidence, Confession, and Admission are dummy variables that describe the evidence in a case; they are as described in the text. Values in parentheses are standard deviations for columns (1)--(4) and (6)--(8), and t-values for two-sided t-tests in columns (5) and (9). ** = significant at 5\% level, * = significant at 10\% level.
\end{tablenotes} 
\end{threeparttable}
}
\end{sideways}
\end{table}

\newpage{}
\begin{table}
[ph!]
\center\resizebox{\textwidth}{!}{%
\begin{threeparttable}\caption{Effect of Strike Exhaustion on Conviction\label{tab:Defendant-Pl-Strike-Exhaustion}}

\begin{tabular}{lr@{\extracolsep{0pt}.}lr@{\extracolsep{0pt}.}lr@{\extracolsep{0pt}.}l|r@{\extracolsep{0pt}.}lr@{\extracolsep{0pt}.}lr@{\extracolsep{0pt}.}l}
\hline 
 &
\multicolumn{12}{c}{}\tabularnewline
 &
\multicolumn{6}{c|}{Defense} &
\multicolumn{6}{c}{Prosecution}\tabularnewline
\cline{2-13}
 &
\multicolumn{2}{c}{Guilty} &
\multicolumn{2}{c}{Guilty} &
\multicolumn{2}{c|}{Guilty} &
\multicolumn{2}{c}{Guilty} &
\multicolumn{2}{c}{Guilty} &
\multicolumn{2}{c}{Guilty}\tabularnewline
 &
\multicolumn{2}{c}{(1)} &
\multicolumn{2}{c}{(2)} &
\multicolumn{2}{c|}{(3)} &
\multicolumn{2}{c}{(4)} &
\multicolumn{2}{c}{(5)} &
\multicolumn{2}{c}{(6)}\tabularnewline
\hline 
\hline 
Exhausts Strikes &
0&12{*}{*} &
0&11{*}{*} &
0&12{*}{*} &
0&02 &
0&01 &
-0&01\tabularnewline
 &
(0&05) &
(0&05) &
(0&05) &
(0&06) &
(0&06) &
(0&06)\tabularnewline
n/n-1 Group &
-0&01 &
0&01 &
0&00 &
0&05 &
0&04 &
0&01\tabularnewline
 &
(0&05) &
(0&05) &
(0&05) &
(0&05) &
(0&05) &
(0&05)\tabularnewline
Felony &
0&30{*}{*}{*} &
0&30{*}{*}{*} &
0&24{*}{*}{*} &
0&29{*}{*}{*} &
0&28{*}{*}{*} &
0&21{*}{*}{*}\tabularnewline
 &
(0&04) &
(0&04) &
(0&06) &
(0&04) &
(0&04) &
(0&06)\tabularnewline
Life Eligible Crime &
0&03 &
0&03 &
-0&04 &
0&03 &
0&03 &
-0&06\tabularnewline
 &
(0&05) &
(0&05) &
(0&06) &
(0&05) &
(0&05) &
(0&07)\tabularnewline
Constant &
0&29{*}{*}{*} &
0&31{*}{*}{*} &
0&27{*}{*}{*} &
0&31{*}{*}{*} &
0&35{*}{*}{*} &
0&26{*}{*}{*}\tabularnewline
 &
(0&04) &
(0&09) &
(0&09) &
(0&04) &
(0&09) &
(0&09)\tabularnewline
\hline 
Observations &
\multicolumn{2}{c}{567} &
\multicolumn{2}{c}{559} &
\multicolumn{2}{c|}{559} &
\multicolumn{2}{c}{567} &
\multicolumn{2}{c}{559} &
\multicolumn{2}{c}{559}\tabularnewline
\hline 
Adj. R-squared &
0&10 &
0&10 &
0&10 &
0&09 &
0&09 &
0&10\tabularnewline
Defendant Controls &
\multicolumn{2}{c}{NO} &
\multicolumn{2}{c}{YES} &
\multicolumn{2}{c|}{YES} &
\multicolumn{2}{c}{NO} &
\multicolumn{2}{c}{YES} &
\multicolumn{2}{c}{YES}\tabularnewline
Case Controls &
\multicolumn{2}{c}{NO} &
\multicolumn{2}{c}{NO} &
\multicolumn{2}{c|}{YES} &
\multicolumn{2}{c}{NO} &
\multicolumn{2}{c}{NO} &
\multicolumn{2}{c}{YES}\tabularnewline
\hline 
\end{tabular}

\begin{tablenotes} 
\item\emph{Notes}: This table presents OLS regressions of peremptory strike exhaustion on guilt. The dependent variable in all specifications is the proportion of charges pending as of the date of jury selection on which a defendant is eventually convicted. The primary coefficient of interest is Exhausts Strikes, which = 1 if the party used up its $n$ strikes and = 0 if s/he used fewer ($<n$) strikes. n/n-1 Group = 1 if the party used either $n$ or $n-1$ strikes and = 0 otherwise. Columns (1)--(3) measure defense exhaustion of peremptory strikes; columns (4)--(6) measure prosecution exhaustion of those strikes. Defendant Controls include defendant age and years of previous imprisonment in Florida state prison, and separate dummy variables for whether a defendant is female, Black or Hispanic. Case Controls include number of for cause strikes issued by the judge, prosecution, and defense; total number of counts charged; a dummy for whether the opposing party exhausted its peremptory strikes; and a fixed effect for the year in which jury selection occurred. Heteroskedastic robust standard errors are reported in parentheses. *** = significant at 1\% level, ** = significant at 5\% level, * = significant at 10\% level.
\end{tablenotes} 
\end{threeparttable}
}
\end{table}

\newpage{}
\begin{table}
[ph!]
\noindent\resizebox{\textwidth}{!}{%
\begin{threeparttable}\caption{Effect of Strike Exhaustion on Conviction for Black and Non-Black
Defendants\label{tab:Defendant-Bl-NBl-Exhaustion}}

\begin{tabular}{lr@{\extracolsep{0pt}.}lr@{\extracolsep{0pt}.}lr@{\extracolsep{0pt}.}l|r@{\extracolsep{0pt}.}lr@{\extracolsep{0pt}.}lr@{\extracolsep{0pt}.}l}
\hline 
 &
\multicolumn{12}{c}{}\tabularnewline
 &
\multicolumn{6}{c|}{Black Defendants} &
\multicolumn{6}{c}{Non-Black Defendants}\tabularnewline
\cline{2-13}
 &
\multicolumn{2}{c}{Guilty} &
\multicolumn{2}{c}{Guilty} &
\multicolumn{2}{c|}{Guilty} &
\multicolumn{2}{c}{Guilty} &
\multicolumn{2}{c}{Guilty} &
\multicolumn{2}{c}{Guilty}\tabularnewline
 &
\multicolumn{2}{c}{(1)} &
\multicolumn{2}{c}{(2)} &
\multicolumn{2}{c|}{(3)} &
\multicolumn{2}{c}{(4)} &
\multicolumn{2}{c}{(5)} &
\multicolumn{2}{c}{(6)}\tabularnewline
\hline 
\hline 
Exhausts Strikes &
0&19{*}{*}{*} &
0&20{*}{*}{*} &
0&18{*}{*}{*} &
0&06 &
0&04 &
0&07\tabularnewline
 &
(0&07) &
(0&07) &
(0&07) &
(0&07) &
(0&07) &
(0&08)\tabularnewline
n/n-1 Group &
-0&07 &
-0&06 &
-0&07 &
0&07 &
0&08 &
0&08\tabularnewline
 &
(0&06) &
(0&07) &
(0&07) &
(0&07) &
(0&07) &
(0&07)\tabularnewline
Felony &
0&35{*}{*}{*} &
0&34{*}{*}{*} &
0&31{*}{*}{*} &
0&28{*}{*}{*} &
0&27{*}{*}{*} &
0&16{*}\tabularnewline
 &
(0&06) &
(0&06) &
(0&08) &
(0&06) &
(0&06) &
(0&09)\tabularnewline
Life Eligible Crime &
0&02 &
0&03 &
0&01 &
0&05 &
0&06 &
-0&10\tabularnewline
 &
(0&07) &
(0&07) &
(0&08) &
(0&08) &
(0&08) &
(0&10)\tabularnewline
Constant &
0&25{*}{*}{*} &
0&15 &
0&12 &
0&31{*}{*}{*} &
0&40{*}{*}{*} &
0&36{*}{*}{*}\tabularnewline
 &
(0&06) &
(0&10) &
(0&11) &
(0&06) &
(0&11) &
(0&13)\tabularnewline
\hline 
Observations &
\multicolumn{2}{c}{310} &
\multicolumn{2}{c}{306} &
\multicolumn{2}{c|}{306} &
\multicolumn{2}{c}{257} &
\multicolumn{2}{c}{253} &
\multicolumn{2}{c}{253}\tabularnewline
\hline 
Adj.\ R-squared &
0&12 &
0&12 &
0&11 &
0&08 &
0&07 &
0&09\tabularnewline
Defendant Controls &
\multicolumn{2}{c}{NO} &
\multicolumn{2}{c}{YES} &
\multicolumn{2}{c|}{YES} &
\multicolumn{2}{c}{NO} &
\multicolumn{2}{c}{YES} &
\multicolumn{2}{c}{YES}\tabularnewline
Case Controls &
\multicolumn{2}{c}{NO} &
\multicolumn{2}{c}{NO} &
\multicolumn{2}{c|}{YES} &
\multicolumn{2}{c}{NO} &
\multicolumn{2}{c}{NO} &
\multicolumn{2}{c}{YES}\tabularnewline
\hline 
\end{tabular}

\begin{tablenotes} 
\item\emph{Notes}: This table presents OLS regressions of defendant peremptory strike exhaustion on guilt, split by defendant race. The dependent variable in all specifications is the proportion of charges pending as of the date of jury selection on which a defendant is eventually convicted. Columns (1)--(3) measure exhaustion of peremptory strikes by Black defendants; columns (4)--(6) measure exhaustion of strikes by non-Black defendants. The primary coefficient of interest is Exhausts Strikes, which = 1 if the defendant used up its $n$ strikes and = 0 if s/he used fewer ($<n$) strikes. n/n-1 Group = 1 if the defendant used either $n$ or $n-1$ strikes and = 0 otherwise. Defendant Controls are as described in Table 4, except a dummy variable for whether a defendant is Hispanic is only used in cols. (5)--(6). Case Controls are as described in Table 4. Heteroskedastic robust standard errors are reported in parentheses. *** = significant at 1\% level, ** = significant at 5\% level, * = significant at 10\% level.
\end{tablenotes} 
\end{threeparttable}
}
\end{table}

\newpage{}
\begin{table}
[ph!]
\noindent\resizebox{\textwidth}{!}{%
\begin{threeparttable}\caption{Effect of Defendant Strike Exhaustion With Jury Composition Controls\label{tab:Defendant-All-Bl-Jury-Composition}}

\begin{tabular}{lr@{\extracolsep{0pt}.}lr@{\extracolsep{0pt}.}lr@{\extracolsep{0pt}.}l|r@{\extracolsep{0pt}.}lr@{\extracolsep{0pt}.}lr@{\extracolsep{0pt}.}l}
\hline 
 &
\multicolumn{12}{c}{}\tabularnewline
 &
\multicolumn{6}{c|}{All Defendants} &
\multicolumn{6}{c}{Black Defendants}\tabularnewline
\cline{2-13}
 &
\multicolumn{2}{c}{Guilty} &
\multicolumn{2}{c}{Guilty} &
\multicolumn{2}{c|}{Guilty} &
\multicolumn{2}{c}{Guilty} &
\multicolumn{2}{c}{Guilty} &
\multicolumn{2}{c}{Guilty}\tabularnewline
 &
\multicolumn{2}{c}{(1)} &
\multicolumn{2}{c}{(2)} &
\multicolumn{2}{c|}{(3)} &
\multicolumn{2}{c}{(4)} &
\multicolumn{2}{c}{(5)} &
\multicolumn{2}{c}{(6)}\tabularnewline
\hline 
\hline 
Exhausts Strikes &
0&12{*}{*} &
0&11{*}{*} &
0&11{*}{*} &
0&20{*}{*}{*} &
0&20{*}{*}{*} &
0&19{*}{*}{*}\tabularnewline
 &
(0&05) &
(0&05) &
(0&05) &
(0&07) &
(0&07) &
(0&07)\tabularnewline
n/n-1 Group &
0&01 &
0&02 &
0&02 &
-0&07 &
-0&06 &
-0&06\tabularnewline
 &
(0&05) &
(0&05) &
(0&05) &
(0&06) &
(0&06) &
(0&07)\tabularnewline
Felony &
0&30{*}{*}{*} &
0&29{*}{*}{*} &
0&22{*}{*}{*} &
0&34{*}{*}{*} &
0&33{*}{*}{*} &
0&30{*}{*}{*}\tabularnewline
 &
(0&04) &
(0&04) &
(0&06) &
(0&06) &
(0&06) &
(0&08)\tabularnewline
Life Eligible Crime &
0&03 &
0&04 &
-0&04 &
0&02 &
0&03 &
0&01\tabularnewline
 &
(0&05) &
(0&05) &
(0&06) &
(0&07) &
(0&07) &
(0&08)\tabularnewline
Prop. Black in Pool &
0&46 &
0&48 &
0&49 &
0&71 &
0&77{*} &
0&82{*}\tabularnewline
 &
(0&32) &
(0&32) &
(0&32) &
(0&45) &
(0&45) &
(0&45)\tabularnewline
Avg. Age in Pool &
0&00 &
0&00 &
0&00 &
-0&00 &
-0&00 &
-0&00\tabularnewline
 &
(0&01) &
(0&01) &
(0&01) &
(0&01) &
(0&01) &
(0&01)\tabularnewline
Prop. Female in Pool &
0&09 &
0&14 &
0&14 &
0&26 &
0&31 &
0&30\tabularnewline
 &
(0&20) &
(0&20) &
(0&20) &
(0&28) &
(0&28) &
(0&29)\tabularnewline
Ln Median Pool Income &
0&35 &
0&34 &
0&37{*} &
0&51{*} &
0&54{*} &
0&55{*}\tabularnewline
 &
(0&22) &
(0&22) &
(0&23) &
(0&30) &
(0&30) &
(0&31)\tabularnewline
Prop. Democrat in Pool &
-0&29 &
-0&30 &
-0&30 &
-0&10 &
-0&03 &
-0&04\tabularnewline
 &
(0&25) &
(0&25) &
(0&25) &
(0&32) &
(0&33) &
(0&35)\tabularnewline
Prop. Republican in Pool &
0&40 &
0&38 &
0&37 &
0&62{*} &
0&67{*}{*} &
0&65{*}\tabularnewline
 &
(0&25) &
(0&25) &
(0&26) &
(0&33) &
(0&33) &
(0&35)\tabularnewline
Constant &
-3&78 &
-3&78 &
-4&17{*} &
-5&58{*} &
-6&21{*} &
-6&27{*}\tabularnewline
 &
(2&46) &
(2&49) &
(2&53) &
(3&27) &
(3&33) &
(3&44)\tabularnewline
\hline 
Observations &
\multicolumn{2}{c}{567} &
\multicolumn{2}{c}{559} &
\multicolumn{2}{c|}{559} &
\multicolumn{2}{c}{310} &
\multicolumn{2}{c}{306} &
\multicolumn{2}{c}{306}\tabularnewline
\hline 
Adj. R-squared &
0&10 &
0&11 &
0&11 &
0&13 &
0&14 &
0&13\tabularnewline
Defendant Controls &
\multicolumn{2}{c}{NO} &
\multicolumn{2}{c}{YES} &
\multicolumn{2}{c|}{YES} &
\multicolumn{2}{c}{NO} &
\multicolumn{2}{c}{YES} &
\multicolumn{2}{c}{YES}\tabularnewline
Case Controls &
\multicolumn{2}{c}{NO} &
\multicolumn{2}{c}{NO} &
\multicolumn{2}{c|}{YES} &
\multicolumn{2}{c}{NO} &
\multicolumn{2}{c}{NO} &
\multicolumn{2}{c}{YES}\tabularnewline
\hline 
\end{tabular}

\begin{tablenotes} 
\item\emph{Notes}: This table presents OLS regressions of defendant peremptory strike exhaustion on guilt, with controls for jury pool composition. The dependent variable in all specifications is the proportion of charges pending as of the date of jury selection on which a defendant is eventually convicted. Columns (1)--(3) include all defendants; columns (4)--(6) are limited to Black defendants. Exhausts Strikes and n/n-1 Group are as described in Table 5. Defendant Controls are as described in Table 4, except dummy variables for whether a defendant is Black or Hispanic are only used in cols. (2)--(3). Case Controls are as described in Table 4. The various jury composition controls are as described in the text. Heteroskedastic robust standard errors are reported in parentheses. *** = significant at 1\% level, ** = significant at 5\% level, * = significant at 10\% level.
\end{tablenotes} 
\end{threeparttable}
}
\end{table}
\newpage{}
\begin{table}
[ph!]
\noindent\resizebox{\textwidth}{!}{%
\begin{threeparttable}\caption{Effect of Defendant Strike Exhaustion on Conviction: Attorney Controls\label{tab:Defendant-Strike-Atty}}

\begin{tabular}{lr@{\extracolsep{0pt}.}lr@{\extracolsep{0pt}.}lr@{\extracolsep{0pt}.}lr@{\extracolsep{0pt}.}l|r@{\extracolsep{0pt}.}lr@{\extracolsep{0pt}.}lr@{\extracolsep{0pt}.}lr@{\extracolsep{0pt}.}l}
\hline 
 &
\multicolumn{16}{c}{}\tabularnewline
 &
\multicolumn{8}{c|}{All Defendants} &
\multicolumn{8}{c}{Black Defendants}\tabularnewline
\cline{2-17}
 &
\multicolumn{2}{c}{Guilty} &
\multicolumn{2}{c}{Guilty} &
\multicolumn{2}{c}{Guilty} &
\multicolumn{2}{c|}{Guilty} &
\multicolumn{2}{c}{Guilty} &
\multicolumn{2}{c}{Guilty} &
\multicolumn{2}{c}{Guilty} &
\multicolumn{2}{c}{Guilty}\tabularnewline
 &
\multicolumn{2}{c}{(1)} &
\multicolumn{2}{c}{(2)} &
\multicolumn{2}{c}{(3)} &
\multicolumn{2}{c|}{(4)} &
\multicolumn{2}{c}{(5)} &
\multicolumn{2}{c}{(6)} &
\multicolumn{2}{c}{(7)} &
\multicolumn{2}{c}{(8)}\tabularnewline
\hline 
\hline 
Exhausts Strikes &
0&15{*}{*}{*} &
0&17{*}{*}{*} &
0&16{*}{*}{*} &
0&21{*}{*}{*} &
0&19{*}{*}{*} &
0&23{*}{*}{*} &
0&21{*}{*}{*} &
0&25{*}{*}{*}\tabularnewline
 &
(0&05) &
(0&06) &
(0&06) &
(0&07) &
(0&07) &
(0&07) &
(0&08) &
(0&09)\tabularnewline
n/n-1 Group &
-0&04 &
-0&06 &
-0&03 &
-0&08 &
-0&10 &
-0&14{*} &
-0&08 &
-0&13\tabularnewline
 &
(0&05) &
(0&05) &
(0&06) &
(0&07) &
(0&07) &
(0&07) &
(0&09) &
(0&09)\tabularnewline
Felony &
0&19{*}{*}{*} &
0&21{*}{*}{*} &
0&20{*}{*} &
0&19{*}{*} &
0&27{*}{*}{*} &
0&29{*}{*}{*} &
0&26{*}{*} &
0&27{*}{*}\tabularnewline
 &
(0&06) &
(0&06) &
(0&08) &
(0&08) &
(0&09) &
(0&09) &
(0&11) &
(0&12)\tabularnewline
Life Eligible Crime &
-0&06 &
-0&04 &
0&03 &
0&03 &
-0&04 &
-0&01 &
0&08 &
0&10\tabularnewline
 &
(0&06) &
(0&06) &
(0&08) &
(0&08) &
(0&08) &
(0&08) &
(0&09) &
(0&09)\tabularnewline
Constant &
0&30{*}{*}{*} &
0&32{*}{*} &
0&29{*}{*} &
0&34{*}{*} &
0&18 &
0&17 &
0&12 &
0&14\tabularnewline
 &
(0&11) &
(0&13) &
(0&12) &
(0&14) &
(0&15) &
(0&16) &
(0&15) &
(0&17)\tabularnewline
\hline 
Observations &
\multicolumn{2}{c}{527} &
\multicolumn{2}{c}{527} &
\multicolumn{2}{c}{393} &
\multicolumn{2}{c|}{393} &
\multicolumn{2}{c}{293} &
\multicolumn{2}{c}{293} &
\multicolumn{2}{c}{227} &
\multicolumn{2}{c}{227}\tabularnewline
\hline 
Adj.\ R-squared &
0&14 &
0&13 &
0&13 &
0&15 &
0&15 &
0&16 &
0&13 &
0&14\tabularnewline
Defendant \& Case Ctrls &
\multicolumn{2}{c}{YES} &
\multicolumn{2}{c}{YES} &
\multicolumn{2}{c}{YES} &
\multicolumn{2}{c|}{YES} &
\multicolumn{2}{c}{YES} &
\multicolumn{2}{c}{YES} &
\multicolumn{2}{c}{YES} &
\multicolumn{2}{c}{YES}\tabularnewline
Attorney Controls &
\multicolumn{2}{c}{YES} &
\multicolumn{2}{c}{YES} &
\multicolumn{2}{c}{YES} &
\multicolumn{2}{c|}{YES} &
\multicolumn{2}{c}{YES} &
\multicolumn{2}{c}{YES} &
\multicolumn{2}{c}{YES} &
\multicolumn{2}{c}{YES}\tabularnewline
Both Sides &
\multicolumn{2}{c}{NO} &
\multicolumn{2}{c}{YES} &
\multicolumn{2}{c}{NO} &
\multicolumn{2}{c|}{YES} &
\multicolumn{2}{c}{NO} &
\multicolumn{2}{c}{YES} &
\multicolumn{2}{c}{NO} &
\multicolumn{2}{c}{YES}\tabularnewline
PD Only &
\multicolumn{2}{c}{NO} &
\multicolumn{2}{c}{NO} &
\multicolumn{2}{c}{YES} &
\multicolumn{2}{c|}{YES} &
\multicolumn{2}{c}{NO} &
\multicolumn{2}{c}{NO} &
\multicolumn{2}{c}{YES} &
\multicolumn{2}{c}{YES}\tabularnewline
\hline 
\end{tabular}

\begin{tablenotes} 
\item\emph{Notes}: This table presents OLS regressions of defendant peremptory strike exhaustion on guilt, with controls for attorney characteristics. The dependent variable in all specifications is the proportion of charges pending as of the date of jury selection on which a defendant is eventually convicted. Columns (1)-(4) include all defendants; columns (5)-(8) are limited to Black defendants. Exhausts Strikes and n/n-1 Group are as described in Table 5. Defendant \& Case Controls are as described in Table 4 (except dummy variables for whether a defendant is Black or Hispanic are only used in cols. (1)--(4)). Attorney Controls include separate controls for both prosecutor and defense attorney years of experience and ranking of law school attended. Both Sides controls for the number of other cases in our sample in which a prosecutor or defense attorney used up her peremptory strikes, and the number of other cases in our sample in which she used one less strike than the limit. PD Only limits the sample only to cases involving public defenders; the other columns include a dummy variable for whether the defense attorney was a public defender. Heteroskedastic robust standard errors are reported in parentheses. *** = significant at 1\% level, ** = significant at 5\% level, * = significant at 10\% level.
\end{tablenotes} 
\end{threeparttable}
}
\end{table}

\newpage{}
\begin{table}
[ph!]
\noindent\resizebox{.96\textwidth}{!}{%
\begin{threeparttable}\caption{Effect of Defendant Strike Exhaustion on Conviction: Charge/Judge
Fixed Effects\label{tab:Defendant-Strike-Year-Charge-Judge-FEs}}

\begin{tabular}{lr@{\extracolsep{0pt}.}lr@{\extracolsep{0pt}.}lr@{\extracolsep{0pt}.}l|r@{\extracolsep{0pt}.}lr@{\extracolsep{0pt}.}lr@{\extracolsep{0pt}.}l}
\hline 
 &
\multicolumn{12}{c}{}\tabularnewline
 &
\multicolumn{6}{c|}{All Defendants} &
\multicolumn{6}{c}{Black Defendants}\tabularnewline
\cline{2-13}
 &
\multicolumn{2}{c}{Guilty} &
\multicolumn{2}{c}{Guilty} &
\multicolumn{2}{c|}{Guilty} &
\multicolumn{2}{c}{Guilty} &
\multicolumn{2}{c}{Guilty} &
\multicolumn{2}{c}{Guilty}\tabularnewline
 &
\multicolumn{2}{c}{(1)} &
\multicolumn{2}{c}{(2)} &
\multicolumn{2}{c|}{(3)} &
\multicolumn{2}{c}{(4)} &
\multicolumn{2}{c}{(5)} &
\multicolumn{2}{c}{(6)}\tabularnewline
\hline 
\hline 
Exhausts Strikes &
0&12{*}{*} &
0&12{*}{*} &
0&16{*}{*}{*} &
0&18{*}{*} &
0&19{*}{*} &
0&21{*}{*}{*}\tabularnewline
 &
(0&05) &
(0&06) &
(0&06) &
(0&07) &
(0&08) &
(0&08)\tabularnewline
n/n-1 Group &
-0&00 &
0&02 &
-0&04 &
-0&08 &
-0&05 &
-0&12\tabularnewline
 &
(0&05) &
(0&05) &
(0&05) &
(0&07) &
(0&08) &
(0&08)\tabularnewline
Felony &
0&17{*}{*}{*} &
0&13 &
-0&01 &
0&22{*}{*} &
0&35 &
0&12\tabularnewline
 &
(0&06) &
(0&28) &
(0&27) &
(0&09) &
(0&32) &
(0&32)\tabularnewline
Life Eligible Crime &
-0&03 &
-0&04 &
-0&07 &
0&08 &
0&02 &
0&03\tabularnewline
 &
(0&07) &
(0&07) &
(0&07) &
(0&09) &
(0&09) &
(0&09)\tabularnewline
Constant &
0&44{*}{*}{*} &
0&20 &
0&37{*} &
0&13 &
0&09 &
0&08\tabularnewline
 &
(0&14) &
(0&15) &
(0&21) &
(0&20) &
(0&21) &
(0&30)\tabularnewline
\hline 
Observations &
\multicolumn{2}{c}{552} &
\multicolumn{2}{c}{552} &
\multicolumn{2}{c|}{520} &
\multicolumn{2}{c}{303} &
\multicolumn{2}{c}{303} &
\multicolumn{2}{c}{290}\tabularnewline
\hline 
Adj.\ R-squared &
0&11 &
0&10 &
0&15 &
0&13 &
0&09 &
0&14\tabularnewline
Defendant \& Case Ctrls &
\multicolumn{2}{c}{YES} &
\multicolumn{2}{c}{YES} &
\multicolumn{2}{c|}{YES} &
\multicolumn{2}{c}{YES} &
\multicolumn{2}{c}{YES} &
\multicolumn{2}{c}{YES}\tabularnewline
Charge Fixed Effect &
\multicolumn{2}{c}{YES} &
\multicolumn{2}{c}{NO} &
\multicolumn{2}{c|}{YES} &
\multicolumn{2}{c}{YES} &
\multicolumn{2}{c}{NO} &
\multicolumn{2}{c}{YES}\tabularnewline
Judge Fixed Effect &
\multicolumn{2}{c}{NO} &
\multicolumn{2}{c}{YES} &
\multicolumn{2}{c|}{YES} &
\multicolumn{2}{c}{NO} &
\multicolumn{2}{c}{YES} &
\multicolumn{2}{c}{YES}\tabularnewline
Attorney Controls &
\multicolumn{2}{c}{NO} &
\multicolumn{2}{c}{NO} &
\multicolumn{2}{c|}{YES} &
\multicolumn{2}{c}{NO} &
\multicolumn{2}{c}{NO} &
\multicolumn{2}{c}{YES}\tabularnewline
\hline 
\end{tabular}

\begin{tablenotes} 
\item\emph{Notes}: This table presents OLS regressions of defendant exhaustion of peremptory strikes on guilt, with the inclusion of various fixed effects. The dependent variable in all specifications is the proportion of charges pending as of the date of jury selection on which a defendant is eventually convicted. Columns (1)--(3) include all defendants; columns (4)--(6) are limited to Black defendants. Exhausts Strikes and n/n-1 Group are as described in Table 5. Defendant \& Case Controls are as described in Table 4, except dummy variables for whether a defendant is Black or Hispanic are only used in cols. (1)--(3). Each specification also includes dummy variables for whether the alleged crime took place in a Correctional Setting, or whether the criminal affidavit described the following evidence against the defendant: Photos/Videos, Recovered Items, Forensic Evidence, Documentation, Admission, and Confession (all described in the text). Charge Fixed Effect controls for offense category (as defined in Anwar et al. 2012) of the first charged offense remaining as of the date of jury selection. Judge Fixed Effect controls for the judge assigned to the case. Attorney Controls include separate controls for both prosecutor and defense attorney years of experience and ranking of law school attended, as well as a dummy variable if the defense attorney was a public defender. Heteroskedastic robust standard errors are reported in parentheses. *** = significant at 1\% level, ** = significant at 5\% level, * = significant at 10\% level. 
\end{tablenotes} 
\end{threeparttable}
}
\end{table}

\newpage{}
\begin{table}
[ph!]
\noindent\resizebox{\textwidth}{!}{%
\begin{threeparttable}\caption{Effect of Defendant Strike Exhaustion on Conviction: Different Strike
Definitions\label{tab:Defendant-Strike-Diff-Defs}}

\begin{tabular}{lr@{\extracolsep{0pt}.}lr@{\extracolsep{0pt}.}lr@{\extracolsep{0pt}.}lr@{\extracolsep{0pt}.}l|r@{\extracolsep{0pt}.}lr@{\extracolsep{0pt}.}lr@{\extracolsep{0pt}.}l}
\hline 
 &
\multicolumn{14}{c}{}\tabularnewline
 &
\multicolumn{8}{c|}{All Defendants} &
\multicolumn{6}{c}{Black Defendants}\tabularnewline
\cline{2-15}
 &
\multicolumn{2}{c}{Guilty} &
\multicolumn{2}{c}{Guilty} &
\multicolumn{2}{c}{Guilty} &
\multicolumn{2}{c|}{Guilty} &
\multicolumn{2}{c}{Guilty} &
\multicolumn{2}{c}{Guilty} &
\multicolumn{2}{c}{Guilty}\tabularnewline
 &
\multicolumn{2}{c}{(1)} &
\multicolumn{2}{c}{(2)} &
\multicolumn{2}{c}{(3)} &
\multicolumn{2}{c|}{(4)} &
\multicolumn{2}{c}{(5)} &
\multicolumn{2}{c}{(6)} &
\multicolumn{2}{c}{(7)}\tabularnewline
\hline 
\hline 
Exhausts Strikes &
0&12{*}{*}{*} &
0&23{*}{*} &
0&21{*}{*} &
0&12{*}{*} &
0&15{*}{*}{*} &
0&20 &
0&26{*}\tabularnewline
 &
(0&04) &
(0&10) &
(0&09) &
(0&05) &
(0&06) &
(0&14) &
(0&15)\tabularnewline
n/n-1 Group &
\multicolumn{2}{c}{\textendash{}} &
-0&25{*} &
-0&23{*}{*} &
-0&02 &
\multicolumn{2}{c}{\textendash{}} &
-0&29 &
-0&38{*}{*}\tabularnewline
 &
\multicolumn{2}{c}{} &
(0&14) &
(0&11) &
(0&06) &
\multicolumn{2}{c}{} &
(0&21) &
(0&16)\tabularnewline
Felony &
0&24{*}{*}{*} &
0&27{*} &
0&15 &
0&25{*}{*}{*} &
0&33{*}{*}{*} &
0&08 &
0&30\tabularnewline
 &
(0&06) &
(0&14) &
(0&12) &
(0&06) &
(0&08) &
(0&19) &
(0&18)\tabularnewline
Life Eligible Crime &
-0&04 &
0&29{*}{*} &
-0&38 &
-0&04 &
0&03 &
\multicolumn{2}{c}{\textendash{}} &
-0&44{*}\tabularnewline
 &
(0&06) &
(0&14) &
(0&24) &
(0&07) &
(0&08) &
\multicolumn{2}{c}{} &
(0&23)\tabularnewline
Constant &
0&34{*}{*}{*} &
0&05 &
0&57{*}{*} &
0&14 &
0&10 &
0&10 &
0&44\tabularnewline
 &
(0&10) &
(0&21) &
(0&24) &
(0&14) &
(0&11) &
(0&26) &
(0&41)\tabularnewline
\hline 
Observations &
\multicolumn{2}{c}{559} &
\multicolumn{2}{c}{92} &
\multicolumn{2}{c}{121} &
\multicolumn{2}{c|}{484} &
\multicolumn{2}{c}{306} &
\multicolumn{2}{c}{51} &
\multicolumn{2}{c}{56}\tabularnewline
\hline 
Adj.\ R-squared &
0&10 &
0&21 &
0&17 &
0&12 &
0&12 &
-0&02 &
0&26\tabularnewline
Defendant \& Case Ctrls &
\multicolumn{2}{c}{YES} &
\multicolumn{2}{c}{YES} &
\multicolumn{2}{c}{YES} &
\multicolumn{2}{c|}{YES} &
\multicolumn{2}{c}{YES} &
\multicolumn{2}{c}{YES} &
\multicolumn{2}{c}{YES}\tabularnewline
Only Black \& White Defs. &
\multicolumn{2}{c}{NO} &
\multicolumn{2}{c}{NO} &
\multicolumn{2}{c}{NO} &
\multicolumn{2}{c|}{YES} &
\multicolumn{2}{c}{NO} &
\multicolumn{2}{c}{NO} &
\multicolumn{2}{c}{NO}\tabularnewline
\# Prosecution Strikes &
\multicolumn{2}{c}{\textendash{}} &
\multicolumn{2}{c}{$n$} &
\multicolumn{2}{c}{$n-1$} &
\multicolumn{2}{c|}{\textendash{}} &
\multicolumn{2}{c}{\textendash{}} &
\multicolumn{2}{c}{$n$} &
\multicolumn{2}{c}{$n-1$}\tabularnewline
\hline 
\end{tabular}

\begin{tablenotes} 
\item\emph{Notes}: This table presents alternate specifications for OLS regressions of defendant peremptory strike exhaustion on guilt. The dependent variable in all specifications is the proportion of charges pending as of the date of jury selection on which a defendant is eventually convicted. Def. Exhausts Strikes is as described before in Table 5, except in columns (1) and (5) this variable = 0 if the defendant used fewer strikes than the limit ($n$ strikes); in all other columns it = 0 if the defendant used just one less strike than the limit ($n-1$ strikes). n/n-1 Group is as described in Table 5. Defendant \& Case Controls are as described in Table 4 (except dummy variables for whether a defendant is Black or Hispanic are only used in cols. (1)--(4) and cols. (1)--(3), respectively). Columns (2), (6), and (3), (7) limit the sample to cases in which the prosecution used $n$ or $n-1$ of its strikes, respectively. Column (4) limits the sample by race to just Black and \white{} defendants. Heteroskedastic robust standard errors are reported in parentheses. *** = significant at 1\% level, ** = significant at 5\% level, * = significant at 10\% level. 
\end{tablenotes} 
\end{threeparttable}
}
\end{table}

\newpage{}
\begin{table}
[ph!]
\begin{sideways}
\noindent\resizebox{.85\textheight}{!}{%
\begin{threeparttable}\caption{Effect of Defendant Strike Exhaustion on Conviction: Verdict-Only
Outcomes\label{tab:Defendant-Strike-Diff-Guilt-Defs}}

\begin{tabular}{lr@{\extracolsep{0pt}.}lr@{\extracolsep{0pt}.}lr@{\extracolsep{0pt}.}lr@{\extracolsep{0pt}.}lr@{\extracolsep{0pt}.}l|r@{\extracolsep{0pt}.}lr@{\extracolsep{0pt}.}lr@{\extracolsep{0pt}.}lr@{\extracolsep{0pt}.}lr@{\extracolsep{0pt}.}l}
\hline 
 &
\multicolumn{6}{c}{} &
\multicolumn{10}{c}{} &
\multicolumn{2}{c}{} &
\multicolumn{2}{c}{}\tabularnewline
 &
\multicolumn{10}{c|}{All Defendants} &
\multicolumn{10}{c}{Black Defendants}\tabularnewline
\cline{2-21}
 &
\multicolumn{2}{c}{Guilty} &
\multicolumn{2}{c}{Guilty} &
\multicolumn{2}{c}{Guilty} &
\multicolumn{2}{c}{Guilty} &
\multicolumn{2}{c|}{Guilty} &
\multicolumn{2}{c}{Guilty} &
\multicolumn{2}{c}{Guilty} &
\multicolumn{2}{c}{Guilty} &
\multicolumn{2}{c}{Guilty} &
\multicolumn{2}{c}{Guilty}\tabularnewline
 &
\multicolumn{2}{c}{(1)} &
\multicolumn{2}{c}{(2)} &
\multicolumn{2}{c}{(3)} &
\multicolumn{2}{c}{(4)} &
\multicolumn{2}{c|}{(5)} &
\multicolumn{2}{c}{(6)} &
\multicolumn{2}{c}{(7)} &
\multicolumn{2}{c}{(8)} &
\multicolumn{2}{c}{(9)} &
\multicolumn{2}{c}{(10)}\tabularnewline
\hline 
\hline 
Exhausts Strikes &
0&11{*} &
0&07 &
0&11{*} &
0&11{*} &
0&13{*}{*} &
0&17{*}{*} &
0&12 &
0&14{*} &
0&20{*}{*} &
0&20{*}{*}{*}\tabularnewline
 &
(0&06) &
(0&06) &
(0&06) &
(0&06) &
(0&06) &
(0&07) &
(0&08) &
(0&08) &
(0&08) &
(0&08)\tabularnewline
n/n-1 Group &
-0&01 &
0&03 &
0&00 &
-0&00 &
-0&05 &
-0&07 &
-0&03 &
-0&05 &
-0&09 &
-0&12\tabularnewline
 &
(0&06) &
(0&06) &
(0&06) &
(0&07) &
(0&06) &
(0&08) &
(0&08) &
(0&08) &
(0&09) &
(0&08)\tabularnewline
Felony &
0&23{*}{*}{*} &
0&25{*}{*}{*} &
0&23{*}{*}{*} &
0&25{*}{*}{*} &
0&25{*}{*}{*} &
0&34{*}{*}{*} &
0&33{*}{*}{*} &
0&34{*}{*}{*} &
0&35{*}{*}{*} &
0&32{*}{*}{*}\tabularnewline
 &
(0&07) &
(0&07) &
(0&07) &
(0&07) &
(0&07) &
(0&09) &
(0&09) &
(0&09) &
(0&10) &
(0&09)\tabularnewline
Life Eligible Crime &
0&02 &
0&04 &
0&00 &
0&03 &
0&06 &
0&09 &
0&06 &
0&05 &
0&08 &
0&12\tabularnewline
 &
(0&07) &
(0&07) &
(0&08) &
(0&09) &
(0&08) &
(0&10) &
(0&10) &
(0&10) &
(0&11) &
(0&10)\tabularnewline
Constant &
0&27{*}{*} &
0&30{*}{*}{*} &
0&24{*}{*} &
0&25{*}{*} &
0&28{*}{*} &
0&12 &
0&13 &
0&11 &
0&08 &
0&13\tabularnewline
 &
(0&11) &
(0&11) &
(0&11) &
(0&12) &
(0&11) &
(0&13) &
(0&13) &
(0&13) &
(0&13) &
(0&13)\tabularnewline
\hline 
Observations &
\multicolumn{2}{c}{491} &
\multicolumn{2}{c}{491} &
\multicolumn{2}{c}{491} &
\multicolumn{2}{c}{434} &
\multicolumn{2}{c|}{491} &
\multicolumn{2}{c}{274} &
\multicolumn{2}{c}{274} &
\multicolumn{2}{c}{274} &
\multicolumn{2}{c}{245} &
\multicolumn{2}{c}{274}\tabularnewline
\hline 
Adj.\ R-squared &
0&09 &
0&11 &
0&10 &
0&11 &
0&12 &
0&11 &
0&12 &
0&10 &
0&12 &
0&14\tabularnewline
Defendant \& Case Ctrls &
\multicolumn{2}{c}{YES} &
\multicolumn{2}{c}{YES} &
\multicolumn{2}{c}{YES} &
\multicolumn{2}{c}{YES} &
\multicolumn{2}{c|}{YES} &
\multicolumn{2}{c}{YES} &
\multicolumn{2}{c}{YES} &
\multicolumn{2}{c}{YES} &
\multicolumn{2}{c}{YES} &
\multicolumn{2}{c}{YES}\tabularnewline
Proportion Guilty &
\multicolumn{2}{c}{YES} &
\multicolumn{2}{c}{NO} &
\multicolumn{2}{c}{NO} &
\multicolumn{2}{c}{NO} &
\multicolumn{2}{c|}{NO} &
\multicolumn{2}{c}{YES} &
\multicolumn{2}{c}{NO} &
\multicolumn{2}{c}{NO} &
\multicolumn{2}{c}{NO} &
\multicolumn{2}{c}{NO}\tabularnewline
Any Guilty Verdict &
\multicolumn{2}{c}{NO} &
\multicolumn{2}{c}{YES} &
\multicolumn{2}{c}{NO} &
\multicolumn{2}{c}{NO} &
\multicolumn{2}{c|}{NO} &
\multicolumn{2}{c}{NO} &
\multicolumn{2}{c}{YES} &
\multicolumn{2}{c}{NO} &
\multicolumn{2}{c}{NO} &
\multicolumn{2}{c}{NO}\tabularnewline
Min Count Verdict &
\multicolumn{2}{c}{NO} &
\multicolumn{2}{c}{NO} &
\multicolumn{2}{c}{YES} &
\multicolumn{2}{c}{NO} &
\multicolumn{2}{c|}{NO} &
\multicolumn{2}{c}{NO} &
\multicolumn{2}{c}{NO} &
\multicolumn{2}{c}{YES} &
\multicolumn{2}{c}{NO} &
\multicolumn{2}{c}{NO}\tabularnewline
No Mixed Verdicts &
\multicolumn{2}{c}{NO} &
\multicolumn{2}{c}{NO} &
\multicolumn{2}{c}{NO} &
\multicolumn{2}{c}{YES} &
\multicolumn{2}{c|}{NO} &
\multicolumn{2}{c}{NO} &
\multicolumn{2}{c}{NO} &
\multicolumn{2}{c}{NO} &
\multicolumn{2}{c}{YES} &
\multicolumn{2}{c}{NO}\tabularnewline
Guilty if Mixed + Jail &
\multicolumn{2}{c}{NO} &
\multicolumn{2}{c}{NO} &
\multicolumn{2}{c}{NO} &
\multicolumn{2}{c}{NO} &
\multicolumn{2}{c|}{YES} &
\multicolumn{2}{c}{NO} &
\multicolumn{2}{c}{NO} &
\multicolumn{2}{c}{NO} &
\multicolumn{2}{c}{NO} &
\multicolumn{2}{c}{YES}\tabularnewline
\hline 
\end{tabular}

\begin{tablenotes} 
\item\emph{Notes}: This table presents alternate specifications for OLS regressions of defendant peremptory strike exhaustion on guilt. All columns are limited to cases in which a jury issued at least one guilty or not guilty verdict; columns (1)-(5) are for all defendants and columns (6)-(10) are limited to Black defendants. Defendant \& Case Controls are as described in Table 4, except dummy variables for whether a defendant is Black or Hispanic are only used in cols. (1)--(5). In columns (1) and (6), the dependent variable is the proportion of charges tried to the jury on which the defendant is found guilty by the jury. In columns (2) and (7), Def. Exhausts Strikes = 1 if the jury issued a guilty verdict on at least count, regardless of what it does on the other counts. In columns (3) and (8), a case is categorized as either guilty or not guilty depending on the verdict the jury issued on the lowest count tried to the jury. In columns (4) and (9), cases with a mixed guilty/not guilty verdict are excluded. In columns (5) and (10), cases with a mixed verdict are categorized as guilty if the defendant received a jail sentence (net of credit for time served); otherwise they are categorized as not guilty. Heteroskedastic robust standard errors are reported in parentheses. *** = significant at 1\% level, ** = significant at 5\% level, * = significant at 10\% level. 
\end{tablenotes} 
\end{threeparttable}
}
\end{sideways}
\end{table}

\newpage{}
\begin{table}
[ph!]
\center\resizebox{0.90\textwidth}{!}{%
\begin{threeparttable}\caption{Effect of Placebo Defendant Strike Exhaustion\label{tab:Defendant-Placebo-Strike-Exhaustion}}

\begin{tabular}{lr@{\extracolsep{0pt}.}lr@{\extracolsep{0pt}.}lr@{\extracolsep{0pt}.}l|r@{\extracolsep{0pt}.}lr@{\extracolsep{0pt}.}lr@{\extracolsep{0pt}.}l}
\hline 
 &
\multicolumn{12}{c}{}\tabularnewline
 &
\multicolumn{6}{c|}{All Defendants} &
\multicolumn{6}{c}{Black Defendants}\tabularnewline
\cline{2-13}
 &
\multicolumn{2}{c}{Guilty} &
\multicolumn{2}{c}{Guilty} &
\multicolumn{2}{c|}{Guilty} &
\multicolumn{2}{c}{Guilty} &
\multicolumn{2}{c}{Guilty} &
\multicolumn{2}{c}{Guilty}\tabularnewline
 &
\multicolumn{2}{c}{(1)} &
\multicolumn{2}{c}{(2)} &
\multicolumn{2}{c|}{(3)} &
\multicolumn{2}{c}{(4)} &
\multicolumn{2}{c}{(5)} &
\multicolumn{2}{c}{(6)}\tabularnewline
\hline 
\hline 
Placebo Exhausts &
-0&07 &
-0&05 &
-0&05 &
-0&09 &
-0&07 &
-0&09\tabularnewline
 &
(0&06) &
(0&06) &
(0&06) &
(0&08) &
(0&08) &
(0&08)\tabularnewline
n-1/n-2 Group &
0&01 &
-0&01 &
-0&00 &
-0&03 &
-0&05 &
-0&04\tabularnewline
 &
(0&05) &
(0&05) &
(0&05) &
(0&07) &
(0&07) &
(0&08)\tabularnewline
Felony &
0&26{*}{*}{*} &
0&26{*}{*}{*} &
0&19{*}{*}{*} &
0&30{*}{*}{*} &
0&29{*}{*}{*} &
0&25{*}{*}{*}\tabularnewline
 &
(0&04) &
(0&04) &
(0&06) &
(0&05) &
(0&06) &
(0&08)\tabularnewline
Life Eligible Crime &
0&01 &
0&01 &
-0&07 &
-0&01 &
-0&00 &
-0&03\tabularnewline
 &
(0&05) &
(0&05) &
(0&06) &
(0&06) &
(0&07) &
(0&08)\tabularnewline
Constant &
0&37{*}{*}{*} &
0&40{*}{*}{*} &
0&34{*}{*}{*} &
0&34{*}{*}{*} &
0&27{*}{*}{*} &
0&19{*}\tabularnewline
 &
(0&03) &
(0&08) &
(0&09) &
(0&05) &
(0&10) &
(0&11)\tabularnewline
\hline 
Observations &
\multicolumn{2}{c}{567} &
\multicolumn{2}{c}{559} &
\multicolumn{2}{c|}{559} &
\multicolumn{2}{c}{310} &
\multicolumn{2}{c}{306} &
\multicolumn{2}{c}{306}\tabularnewline
\hline 
Adj.\ R-squared &
0&09 &
0&09 &
0&09 &
0&11 &
0&11 &
0&11\tabularnewline
Other Case Controls &
\multicolumn{2}{c}{NO} &
\multicolumn{2}{c}{NO} &
\multicolumn{2}{c|}{YES} &
\multicolumn{2}{c}{NO} &
\multicolumn{2}{c}{NO} &
\multicolumn{2}{c}{YES}\tabularnewline
\hline 
\end{tabular}

\begin{tablenotes} 
\item\emph{Notes}: This table presents OLS regressions of a placebo test of defendant peremptory strike exhaustion on guilt. The dependent variable in all specifications is the proportion of charges pending as of the date of jury selection on which a defendant is eventually convicted. The primary coefficient of interest is Placebo Exhausts, which = 1 if the defendant used one less strike than the limit ($n-1$ strikes), and = 0 otherwise. n-1/n-2 Group = 1 if the defendant used either $n-1$ or $n-2$ strikes and = 0 otherwise. Columns (1)--(3) measure placebo exhaustion of peremptory strikes for all defendants; columns (4)--(6) measure placebo exhaustion of strikes by just Black defendants. Defendant \& Case Controls are as described in Table 4, except dummy variables for whether a defendant is Black or Hispanic are only used in cols. (2)--(3). Heteroskedastic robust standard errors are reported in parentheses. *** = significant at 1\% level, ** = significant at 5\% level, * = significant at 10\% level.
\end{tablenotes} 
\end{threeparttable}
}
\end{table}
\newpage{}

\appendix

\section*{Appendix}

\setcounter{figure}{0} \renewcommand{\thefigure}{A.\arabic{figure}}

\setcounter{table}{0} \renewcommand{\thetable}{A.\arabic{table}}

\setcounter{footnote}{0}

\section*{Model}

We now formalize the intuition for our identification strategy and
propose an extended framework in which to discuss assumptions.  To
begin, we parameterize the facts of a case $t$ by a fact index $F^{t}$
taking values in $\left[0,1\right]$, where 0 represents the weakest
case conditions for convicting the defendant, and 1 represents the
strongest.\footnote{In reality, $F^{t}$ would be multi-dimensional, capturing information
such as the evidence against a defendant, including testimonies, as
well as other factors that affect the final verdict, such as the judge.
We simplify to a one-dimensional $F^{t}$ for tractability and because
attorneys also likely act on a summary measure when deciding whether
to strike potential jurors.} We define a \emph{juror predisposition function} (JPF) as a function
$j\colon\left[0,1\right]\rightarrow\left[0,1\right]$ that takes as
input the fact index and gives as output the juror's predisposition
against a defendant with the corresponding fact index.\footnote{Imagine on one extreme the JPF $\underline{j}\left(F^{t}\right)=0$,
which represents a juror who would support a ``not guilty'' verdict
no matter the facts; and on the other extreme the JPF $\overline{j}\left(F^{t}\right)=1$,
which represents a juror who would support a ``guilty'' verdict
no matter the facts. An impartial juror would have an identity JPF:
$I\left(F^{t}\right)=F^{t}$.} We define a \emph{jury pool} of size $P$ for trial $t$ to be a
sequence $j_{1}^{t},\ldots j_{P}^{t}$ of JPFs, and a \emph{seated
jury} of size $S$ to be a subsequence $\left\{ j_{a_{k}}\right\} _{k=1}^{S}$
of length $S$ of the jury pool.

When deciding whether to use a peremptory strike, an attorney must
compare the seated jury at hand with the replacement candidate (the
juror next in line to be considered if a strike is used). An important
feature of our Florida data is that an attorney has the same level
of information on all potential jurors when he makes his strike decision,
including their observable characteristics and responses to questions
they are asked during voir dire.\footnote{This is not the case in some other jurisdictions. For example, in
Middlesex County, New Jersey, attorneys can observe the pool, but
the next juror is revealed and questioned only when it is apparent
that a new juror is needed for the jury box (e.g., a juror in the
box has been excused because of a peremptory strike). In this scenario,
an attorney could make his strike decision solely based on the distribution
of observable characteristics of the pool.} 

After trial, the $S$ jurors render their decision through a unanimous
verdict.\footnote{For simplification, we ignore the possibility of a hung jury in our
model. These comprise a small percentage of cases in our dataset.} We can thus aggregate their individual JPFs into a single verdict
function:
\[
G\left(j_{a_{1}},\ldots,j_{a_{S}};F\right)
\]
$G(\centerdot)$ inputs the individual JPFs of the seated jury, as
well as the facts of the case, and outputs 1 (guilty) or 0 (not guilty).\footnote{Sentencing is decided by the judge, not by the jury. The domain of
$G$ could be generalized to $\left\{ 0,1\right\} ^{k}$, where $k$
is the number of counts the defendant is charged with. For simplicity,
here and in the empirical section, we focus on 0/1 measures of guilty.} Instead of this function, for a specific trial $t$ it is sufficient
for both parties to consider a simpler function $G^{t}\colon\mathbb{R}^{S}\rightarrow\left\{ 0,1\right\} $,
defined as:
\[
G^{t}\left(j_{a_{1}}(F^{t}),\ldots,j_{a_{S}}\left(F^{t}\right)\right).
\]
We assume $G^{t}$ is strictly increasing in all directions.

We can now model the simplest scenario: whether an attorney should
exercise his last peremptory strike after his opponent has already
exhausted all of her strikes. Without loss of generality, suppose
$j_{a_{1}}(F^{t})\le j_{a_{2}}(F^{t})\le\ldots\le j_{a_{S}}(F^{t})$,
and let $r$ be the index number of the replacement juror, which is
non-stochastic because the replacement juror is known. For concreteness,
we focus on the scenario in which the defense attorney has only one
strike remaining. She will use this strike if:
\[
G^{t}\left(j_{a_{1}}(F^{t}),\ldots,j_{a_{S}}\left(F^{t}\right)\right)>G\left(j_{a_{1}}(F^{t}),\ldots,j_{a_{S-1}}\left(F^{t}\right),j_{r}\left(F^{t}\right)\right).
\]
Since $G$ is strictly increasing, an equivalent condition is the
following:
\[
j_{a_{S}}\left(F^{t}\right)>j_{r}\left(F^{t}\right)
\]
We will provide a sufficient condition for the strike decision to
be independent of the facts of the case, $F^{t}$, after the following
definition:
\begin{defn}
We say the jury pool $\left\{ j_{k}\right\} _{k=1}^{P}$ is\emph{
$\boldsymbol{F}$-separated} if the ordering of the JPFs $\left\{ j_{k}\right\} _{k=1}^{P}$
is constant on the set \emph{$\boldsymbol{F}$}: for all $F^{t},F^{t'}\in\boldsymbol{F}$,
$F^{t'}>F^{t}$, $1\le l\le P$, $1\le m\le P$, if $j_{l}\left(F^{t}\right)>j_{m}\left(F^{t}\right)$,
then $j_{l}\left(F^{t'}\right)>j_{m}\left(F^{t}\right).$\footnote{If \emph{$\boldsymbol{F}$} is an interval, the definition is equivalent
to the JPFs not intersecting on \emph{$\boldsymbol{F}$}.}
\end{defn}
\begin{thm}
\label{thm:fact-indep}If a jury pool is \emph{$\boldsymbol{F}$}-separated,
then the decision to strike does not depend on $F^{t}\in\boldsymbol{F}$.
\end{thm}
\begin{proof}
Suppose that a jury pool, $\left\{ j_{k}\right\} _{k=1}^{P}$, is
$\boldsymbol{F}$-separated, and has the seated jury $\left\{ j_{a_{k}}\right\} _{k=1}^{S}$
at the time the defense is deciding whether to use the last strike.
Fix $F^{t},F^{t'}\in\boldsymbol{F}$. Without loss of generality,
suppose that $F^{t}>F^{t'}$ and that $j_{a_{1}}(F^{t})\le j_{a_{2}}(F^{t})\le\ldots\le j_{a_{S}}(F^{t})$.
Let $r$ be the index number of the replacement juror. The defense
will use the strike if $j_{a_{S}}\left(F^{t}\right)>j_{r}\left(F^{t}\right)$,
which is equivalent to $j_{a_{S}}\left(F^{t'}\right)>j_{r}\left(F^{t'}\right)$
because the jury pool is $\boldsymbol{F}$-separated. Hence, the strike
is used (or not used) regardless of the value of the fact index.
\end{proof}
\begin{assumption}
\label{assu:before_n-1}The composition of the seated jury at the
time of the last strike decision of the defense does not depend on
the fact index.
\end{assumption}
Assumption (\ref{assu:before_n-1}) allows us to focus on solving
the end of the peremptory challenge process.\footnote{This assumption could be relaxed by modeling the entire sequence of
challenges. However, this simple framework and Theorem~\ref{thm:fact-indep}
capture the main intuition of the setting.} It follows from Theorem~\ref{thm:fact-indep} and Assumption~(\ref{assu:before_n-1})
that the composition of the final seated jury is independent of the
fact index. This independence has important implications. First,
there should be no systematic differences in $F^{t}$ (or anything
outside of $F^{t}$) between $n$ strike cases and $n-1$ strike cases,
since the decision whether to strike does not depend on the facts
of the case. Second, we can allow for asymmetric information or talent:
Even if some attorneys are more skilled or have more information to
identify $F^{t}$ than other attorneys, this variation does not lead
to a difference in the final seated jury. For this result to hold,
we must assume one of two things. Either all attorneys must correctly
identify $\left\{ j_{k}\right\} _{k=1}^{P}$\textemdash that is, during
voir dire they identify the JPF for all jurors for any set of facts,
even though they might not know $F^{t}$. Alternatively, we can assume
attorneys only correctly identify $F^{t}$ and $\left\{ j_{k}(F^{t})\right\} _{k=1}^{P}$.
That is, instead of assuming correct identification of all JPFs at
all points, it is sufficient to assume correct identification of the
fact index in that trial and the JPFs evaluated at that specific fact
index.

We can now aggregate the above decision process to show what we identify
with the difference in averages for the $n$ strike and $n-1$ strike
cases. Suppose there are $N$ cases where all $n$ strikes were used,
and $N_{-1}$ cases where $n-1$ strikes were used. When we compare
the average conviction rates for the subset of $n$ cases versus the
average for the subset of $n-1$ cases, we get:
\[
\hat{\gamma}=\frac{1}{N}\sum_{i=1}^{N}G\left(J_{b_{j_{1}}},\ldots,J_{b_{j_{k}}};F\right)-\frac{1}{N_{-1}}\sum_{i=1}^{N_{-1}}G\left(J_{r},\ldots,J_{b_{j_{k}}};F\right)
\]
Viewing these cases as an i.i.d.\ sample, the Law of Large Numbers
implies that:
\[
\hat{\gamma}\overset{p}{\rightarrow}\mathrm{E}\left\{ G\left(J_{b_{j_{1}}},\ldots,J_{b_{j_{k}}};F\right)|J_{b_{j_{1}}}>J_{r}\right\} -\mathrm{E}\left\{ G\left(J_{r},\ldots,J_{b_{j_{k}}};F\right)|J_{b_{j_{1}}}\le J_{r}\right\} .
\]

\section*{\newpage Appendix Figures}

\begin{figure}
[ph!]
\caption{For Cause Strikes: Judge Initiated\label{fig:For Cause-Strikes:-Judge}}

\noindent\noindent\resizebox{.99\textwidth}{!}{%

\includegraphics{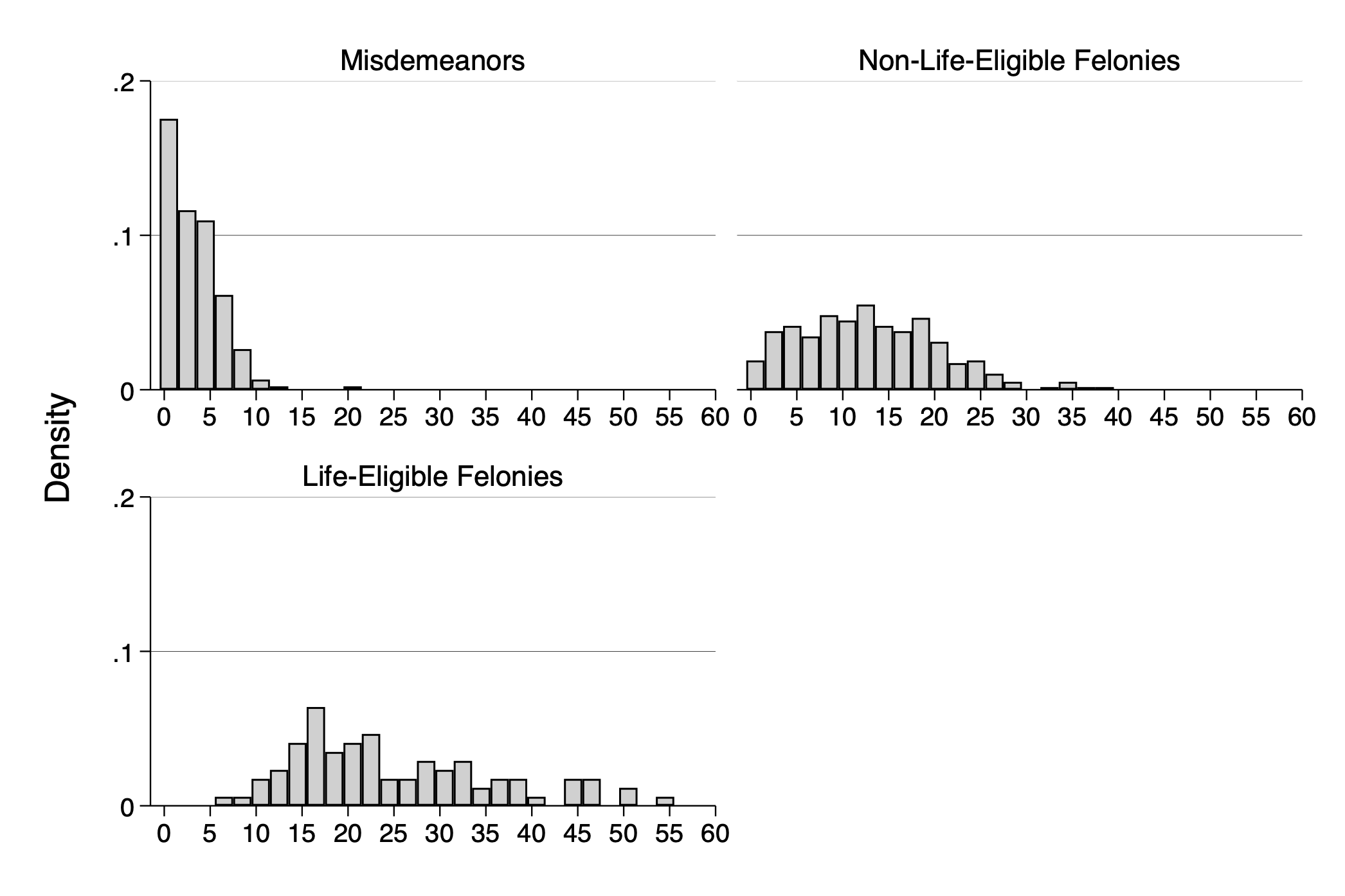}

}

\begin{threeparttable}
\begin{tablenotes} 
\item\emph{Notes}: This figure shows a histogram for the number of for cause strikes initiated by the judge in the 604 trials from 2015 through September 2017 in our full dataset (includes cases in which parties exceeded their peremptory strike limits). Bin size = 2.
\end{tablenotes} 
\end{threeparttable}
\end{figure}

\begin{figure}
[ph!]
\caption{For Cause Strikes: Defense Initiated\label{fig:For Cause-Strikes:-Defense}}

\noindent\noindent\resizebox{.99\textwidth}{!}{%

\includegraphics{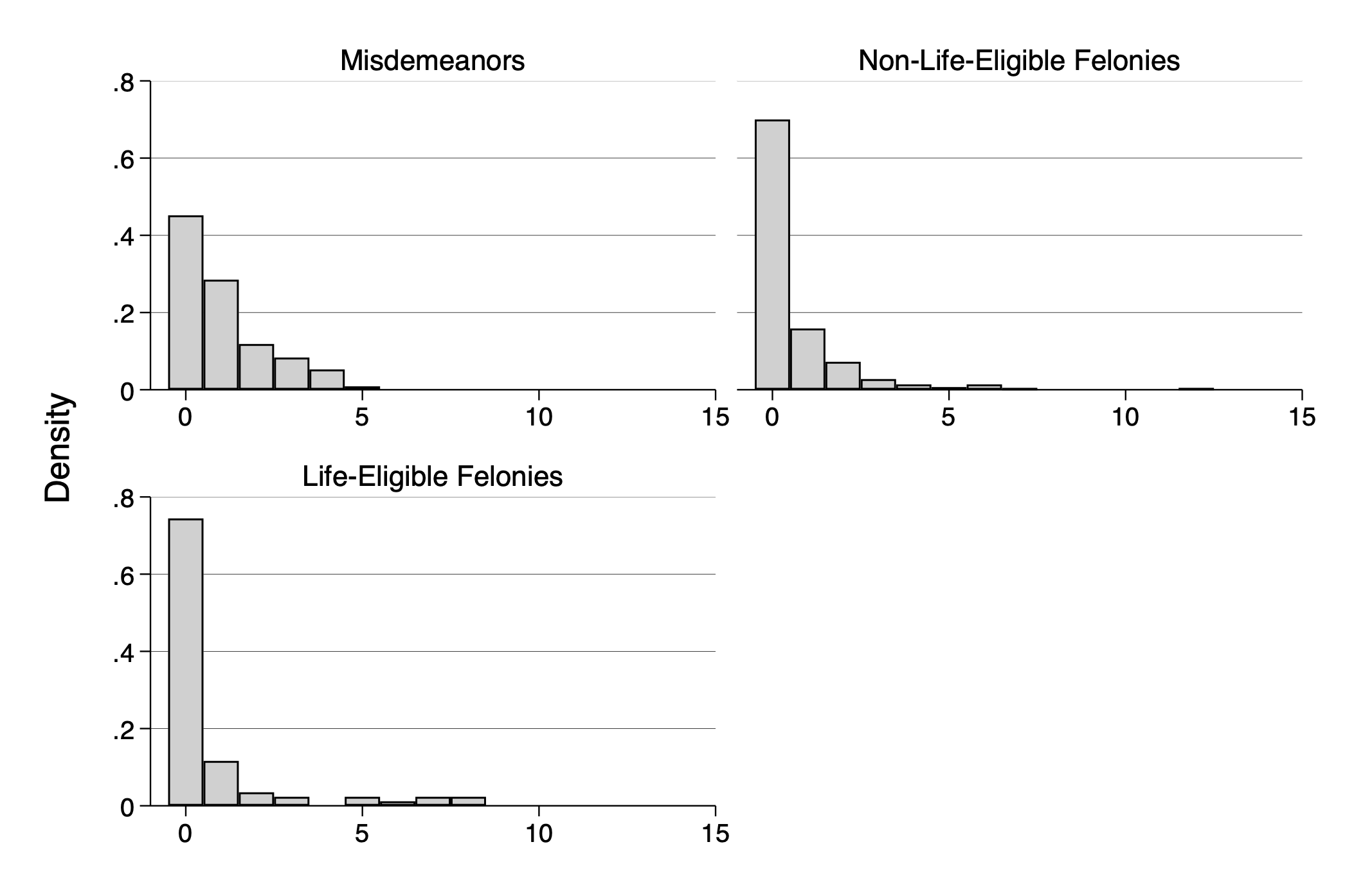}

}

\begin{threeparttable}
\begin{tablenotes} 
\item\emph{Notes}: This figure shows a histogram for the number of for cause strikes successfully requested by the defense in the 604 trials from 2015 through September 2017 in our full dataset (includes cases in which parties exceeded their peremptory strike limits).
\end{tablenotes} 
\end{threeparttable}
\end{figure}

\begin{figure}
[ph!]
\caption{For Cause Strikes: Prosecution Initiated\label{fig:For Cause-Strikes:-Prosecution-1}}

\noindent\noindent\resizebox{.99\textwidth}{!}{%

\includegraphics{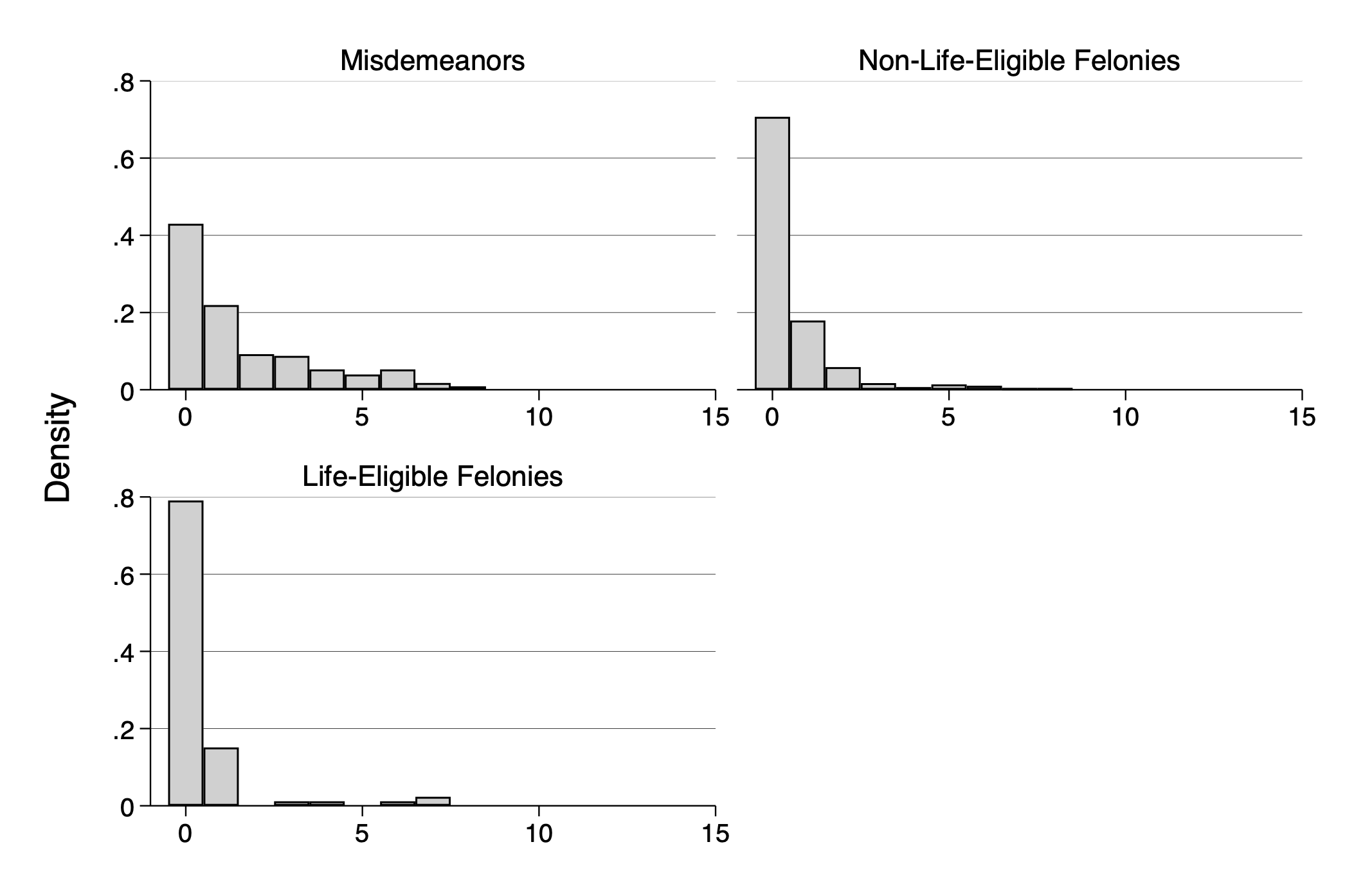}

}

\begin{threeparttable}
\begin{tablenotes} 
\item\emph{Notes}: This figure shows a histogram for the number of for cause strikes successfully requested by the prosecution in the 604 trials from 2015 through September 2017 in our full dataset (includes cases in which parties exceeded their peremptory strike limits).
\end{tablenotes} 
\end{threeparttable}
\end{figure}

\section*{\newpage Appendix Tables}

\vspace{-0.5cm}
\begin{table}
[ph!]
\centering
\resizebox{.90\textwidth}{!}{%
\begin{threeparttable}\caption{Jury Pool Characteristics v. Baseline Characteristics\label{fig:Randomization check}}

\begin{tabular}{lrrrrrr}
\hline 
\multicolumn{5}{l}{} &
 &
\tabularnewline
 &
\multicolumn{4}{c}{Jury Pool Characteristics} &
 &
\tabularnewline
\cline{2-7}
 &
Prop. Black &
Prop. Female &
Avg. Age &
Med. Income &
Prop. Dem. &
Prop. Rep.\tabularnewline
Baseline Characteristics &
(1) &
(2) &
(3) &
(4) &
(5) &
(6)\tabularnewline
\hline 
\hline 
\uline{Defendant characteristics} &
 &
 &
 &
 &
 &
\tabularnewline
Age &
-0.00 &
0.00 &
-0.01 &
3.31 &
-0.00 &
-0.00\tabularnewline
 &
(0.00) &
(0.00) &
(0.01) &
(20.39) &
(0.00) &
(0.00)\tabularnewline
Black &
-0.00 &
0.02 &
-0.59{*} &
-497.89 &
-0.01 &
-0.02{*}{*}\tabularnewline
 &
(0.01) &
(0.01) &
(0.34) &
(535.79) &
(0.01) &
(0.01)\tabularnewline
Hispanic &
-0.00 &
0.01 &
-0.37 &
-427.62 &
-0.01 &
-0.00\tabularnewline
 &
(0.01) &
(0.01) &
(0.44) &
(721.95) &
(0.01) &
(0.01)\tabularnewline
Female &
-0.00 &
0.02 &
-0.98{*}{*} &
-448.22 &
0.01 &
-0.01\tabularnewline
 &
(0.01) &
(0.01) &
(0.43) &
(723.60) &
(0.01) &
(0.01)\tabularnewline
\# Prev. Imprisonments &
-0.00 &
-0.01{*}{*} &
0.16 &
-164.98 &
-0.00 &
0.00\tabularnewline
 &
(0.00) &
(0.00) &
(0.14) &
(249.86) &
(0.00) &
(0.00)\tabularnewline
Years Prev. Imprison. &
0.00 &
-0.00 &
-0.07 &
-22.84 &
-0.00 &
-0.00\tabularnewline
 &
(0.00) &
(0.00) &
(0.06) &
(81.06) &
(0.00) &
(0.00)\tabularnewline
\uline{Attorney characteristics} &
 &
 &
 &
 &
 &
\tabularnewline
Defense Experience &
0.00 &
0.00 &
0.02 &
-50.35{*} &
0.00 &
0.00\tabularnewline
 &
(0.00) &
(0.00) &
(0.02) &
(27.05) &
(0.00) &
(0.00)\tabularnewline
Prosecutor Experience &
-0.00 &
0.00{*}{*} &
-0.05{*}{*} &
-22.22 &
-0.00 &
0.00\tabularnewline
 &
(0.00) &
(0.00) &
(0.02) &
(43.51) &
(0.00) &
(0.00)\tabularnewline
Def. Law Sch. Rank &
-0.00 &
0.00 &
-0.01{*}{*} &
-2.31 &
0.00 &
-0.00\tabularnewline
 &
(0.00) &
(0.00) &
(0.00) &
(3.94) &
(0.00) &
(0.00)\tabularnewline
Pr. Law Sch. Rank &
-0.00 &
-0.00 &
-0.00 &
1.95 &
-0.00 &
0.00\tabularnewline
 &
(0.00) &
(0.00) &
(0.00) &
(4.06) &
(0.00) &
(0.00)\tabularnewline
PD &
-0.00 &
-0.00 &
0.10 &
-278.13 &
0.01 &
0.00\tabularnewline
 &
(0.01) &
(0.01) &
(0.40) &
(625.52) &
(0.01) &
(0.01)\tabularnewline
\uline{Case characteristics} &
 &
 &
 &
 &
 &
\tabularnewline
Homicide &
-0.01 &
-0.00 &
0.07 &
1,690.41 &
-0.00 &
0.00\tabularnewline
 &
(0.01) &
(0.02) &
(0.55) &
(1,161.04) &
(0.02) &
(0.02)\tabularnewline
Other Violent Offense &
-0.00 &
-0.00 &
-0.36 &
99.21 &
0.02{*} &
-0.01\tabularnewline
 &
(0.01) &
(0.01) &
(0.42) &
(647.36) &
(0.01) &
(0.01)\tabularnewline
Property Offense &
-0.01 &
0.02{*} &
-0.19 &
745.09 &
0.02{*}{*} &
-0.01\tabularnewline
 &
(0.01) &
(0.01) &
(0.45) &
(685.94) &
(0.01) &
(0.01)\tabularnewline
Drug Offense &
0.00 &
0.04{*}{*} &
-0.67 &
1,044.59 &
0.03{*} &
-0.02\tabularnewline
 &
(0.01) &
(0.02) &
(0.60) &
(1,000.35) &
(0.01) &
(0.01)\tabularnewline
Sex Offense &
0.01 &
0.02 &
-0.77 &
-306.12 &
0.02{*} &
-0.03{*}{*}\tabularnewline
 &
(0.01) &
(0.02) &
(0.52) &
(996.50) &
(0.01) &
(0.01)\tabularnewline
Weapons Offense &
-0.01 &
0.02 &
-0.17 &
177.79 &
0.03{*} &
0.02\tabularnewline
 &
(0.01) &
(0.02) &
(0.74) &
(1,268.88) &
(0.02) &
(0.02)\tabularnewline
Counts Charged &
-0.00 &
-0.00 &
0.01 &
109.51 &
0.00 &
-0.00\tabularnewline
 &
(0.00) &
(0.00) &
(0.05) &
(80.80) &
(0.00) &
(0.00)\tabularnewline
\hline 
Observations &
527 &
527 &
527 &
527 &
527 &
527\tabularnewline
\hline 
Adj. R-squared &
-0.02 &
0.02 &
0.01 &
-0.01 &
0.00 &
0.01\tabularnewline
F-Statistic &
0.49 &
2.15 &
1.39 &
0.87 &
1.10 &
1.45\tabularnewline
\hline 
\end{tabular}

\begin{tablenotes} 
\item\emph{Notes}: This table shows OLS regression estimates (with heteroskedastic robust standard errors) of various baseline defendant, attorney, and case characteristics on jury pool characteristics as listed at the top of each column. Pool income is the median across jurors within a pool, with each juror's income estimated by the median income in the zip code in which she resides (using U.S. Census data and 2017 inflation-adjusted dollars). "Other Offense" is an omitted category for case characteristics. All variables are as described in previous tables and the text. The F-statistic jointly tests whether all coefficients are = 0 in a given regression. *** = significant at 1\% level, ** = significant at 5\% level, * = significant at 10\% level.
\end{tablenotes} 
\end{threeparttable}
}
\end{table}

\newpage{}

\bibliographystyle{plainnat}
\bibliography{paper}

\begin{thebibliography}{22}
\providecommand{\natexlab}[1]{#1}
\providecommand{\url}[1]{\texttt{#1}}
\expandafter\ifx\csname urlstyle\endcsname\relax
  \providecommand{\doi}[1]{doi: #1}\else
  \providecommand{\doi}{doi: \begingroup \urlstyle{rm}\Url}\fi

\bibitem[Abrams et~al.(2012)Abrams, Bertrand, and
  Mullainathan]{abrams2012judges}
David~S Abrams, Marianne Bertrand, and Sendhil Mullainathan.
\newblock Do judges vary in their treatment of race?
\newblock \emph{The Journal of Legal Studies}, 41\penalty0 (2):\penalty0
  347--383, 2012.

\bibitem[Alesina and La~Ferrara(2014)]{alesina2014test}
Alberto Alesina and Eliana La~Ferrara.
\newblock A test of racial bias in capital sentencing.
\newblock \emph{American Economic Review}, 104\penalty0 (11):\penalty0
  3397--3433, 2014.

\bibitem[Anwar et~al.(2012)Anwar, Bayer, and Hjalmarsson]{anwar2012impact}
Shamena Anwar, Patrick Bayer, and Randi Hjalmarsson.
\newblock The impact of jury race in criminal trials.
\newblock \emph{The Quarterly Journal of Economics}, 127\penalty0 (2):\penalty0
  1017--1055, 2012.

\bibitem[Anwar et~al.(2014)Anwar, Bayer, and Hjalmarsson]{anwar2014role}
Shamena Anwar, Patrick Bayer, and Randi Hjalmarsson.
\newblock The role of age in jury selection and trial outcomes.
\newblock \emph{The Journal of Law and Economics}, 57\penalty0 (4):\penalty0
  1001--1030, 2014.

\bibitem[Anwar et~al.(2022)Anwar, Bayer, and Hjalmarsson]{anwar2022unequal}
Shamena Anwar, Patrick Bayer, and Randi Hjalmarsson.
\newblock Unequal jury representation and its consequences.
\newblock \emph{American Economic Review: Insights}, 4\penalty0 (2):\penalty0
  159--174, 2022.

\bibitem[Arnold et~al.(2018)Arnold, Dobbie, and Yang]{arnold2018racial}
David Arnold, Will Dobbie, and Crystal~S Yang.
\newblock Racial bias in bail decisions.
\newblock \emph{The Quarterly Journal of Economics}, 133\penalty0 (4):\penalty0
  1885--1932, 2018.

\bibitem[Babcock(1974)]{babcock1974voir}
Barbara~Allen Babcock.
\newblock Voir dire: Preserving its wonderful power.
\newblock \emph{Stan L. Rev.}, 27:\penalty0 545, 1974.

\bibitem[Baldus et~al.(2001)Baldus, Woodworth, Zuckerman, and
  Weiner]{baldus2001use}
David~C Baldus, George Woodworth, David Zuckerman, and Neil~Alan Weiner.
\newblock The use of peremptory challenges in capital murder trials: A legal
  and empirical analysis.
\newblock \emph{U. Pa. J. Const. L.}, 3:\penalty0 3, 2001.

\bibitem[Bibas(2004)]{bibas2004plea}
Stephanos Bibas.
\newblock Plea bargaining outside the shadow of trial.
\newblock \emph{Harvard Law Review}, pages 2463--2547, 2004.

\bibitem[Devine et~al.(2001)Devine, Clayton, Dunford, Seying, and
  Pryce]{devine2001jury}
Dennis~J Devine, Laura~D Clayton, Benjamin~B Dunford, Rasmy Seying, and
  Jennifer Pryce.
\newblock Jury decision making: 45 years of empirical research on deliberating
  groups.
\newblock \emph{Psychology, public policy, and law}, 7\penalty0 (3):\penalty0
  622, 2001.

\bibitem[Diamond et~al.(2009)Diamond, Peery, Dolan, and
  Dolan]{diamond2009achieving}
Shari~Seidman Diamond, Destiny Peery, Francis~J Dolan, and Emily Dolan.
\newblock Achieving diversity on the jury: Jury size and the peremptory
  challenge.
\newblock \emph{Journal of Empirical Legal Studies}, 6\penalty0 (3):\penalty0
  425--449, 2009.

\bibitem[Easterbrook(1983)]{easterbrook1983criminal}
Frank~H Easterbrook.
\newblock Criminal procedure as a market system.
\newblock \emph{The Journal of Legal Studies}, 12\penalty0 (2):\penalty0
  289--332, 1983.

\bibitem[Flanagan(2015)]{flanagan2015peremptory}
Francis~X Flanagan.
\newblock Peremptory challenges and jury selection.
\newblock \emph{The Journal of Law and Economics}, 58\penalty0 (2):\penalty0
  385--416, 2015.

\bibitem[Flanagan(2018)]{doi:10.1086/698193}
Francis~X. Flanagan.
\newblock Race, gender, and juries: Evidence from north carolina.
\newblock \emph{The Journal of Law and Economics}, 61\penalty0 (2):\penalty0
  189--214, 2018.
\newblock \doi{10.1086/698193}.

\bibitem[Ford(2010)]{ford2009modeling}
Roger~Allan Ford.
\newblock Modeling the effects of peremptory challenges on jury selection and
  jury verdicts.
\newblock \emph{Geo. Mason L. Rev.}, 17:\penalty0 377, 2010.

\bibitem[Hoekstra and Street(2021)]{hs2021}
Mark Hoekstra and Brittany Street.
\newblock The effect of own-gender jurors on conviction rates.
\newblock \emph{The Journal of Law and Economics}, 64\penalty0 (3):\penalty0
  513--537, 2021.

\bibitem[Howard(2010)]{howard2010taking}
Maureen~A Howard.
\newblock Taking the high road: Why prosecutors should voluntarily waive
  peremptory challenges.
\newblock \emph{Geo. J. Legal Ethics}, 23:\penalty0 369, 2010.

\bibitem[Mnookin and Kornhauser(1978)]{mnookin1978bargaining}
Robert~H Mnookin and Lewis Kornhauser.
\newblock Bargaining in the shadow of the law: The case of divorce.
\newblock \emph{Yale Law Journal}, 88:\penalty0 950, 1978.

\bibitem[Moro and Van~der Linden(2024)]{moro2024exclusion}
Andrea Moro and Martin Van~der Linden.
\newblock Exclusion of extreme jurors and minority representation: The effect
  of jury selection procedures.
\newblock \emph{The Journal of Law and Economics}, 67\penalty0 (2):\penalty0
  295--336, 2024.

\bibitem[Offit(2019)]{offit2018prosecuting}
Anna Offit.
\newblock Prosecuting in the shadow of the jury.
\newblock Available at SSRN: \url{https://ssrn.com/abstract=3225958}, 2019.

\bibitem[Rehavi and Starr(2014)]{rehavi2014racial}
M~Marit Rehavi and Sonja~B Starr.
\newblock Racial disparity in federal criminal sentences.
\newblock \emph{Journal of Political Economy}, 122\penalty0 (6):\penalty0
  1320--1354, 2014.

\bibitem[Zeisel and Diamond(1977)]{zeisel1977effect}
Hans Zeisel and Shari~Seidman Diamond.
\newblock The effect of peremptory challenges on jury and verdict: An
  experiment in a federal district court.
\newblock \emph{Stan. L. Rev.}, 30:\penalty0 491, 1977.

\end{thebibliography}

\end{document}